\documentclass{amsart}
\usepackage{graphicx}
\usepackage{amssymb,amsmath,amsfonts,latexsym,amsthm,mathrsfs}
\usepackage{amsaddr}
\numberwithin{equation}{section}

\newtheorem{theorem}{Theorem}[section]

\newtheorem{lemma}{Lemma}[section]

\newtheorem{corollary}{Corollary}[section]

\renewcommand{\d}{\mathrm{d}}

\newcommand{\bgamma}{\boldsymbol{\gamma}}

\newcommand{\Dc}{\mathcal{D}}

\newcommand{\bM}{\mathbf{M}}
\newcommand{\Ms}{\mathscr{M}}

\renewcommand{\d}{\mathrm{d}}

\newcommand{\Ds}{\mathscr{D}}

\newcommand{\Hc}{\mathcal{H}}

\newcommand{\Hs}{\mathscr{H}}

\newcommand{\Kc}{\mathcal{K}}

\newcommand{\Mb}{\mathbf{M}}

\newcommand{\Qt}{\tilde{Q}}

\newcommand{\Rc}{\mathcal{R}}

\newcommand{\Vc}{\mathcal{V}}

\newcommand{\bcdot}{\mbox{\boldmath${\cdot}$}}

\newcommand{\btimes}{\mbox{\boldmath${\times}$}}

\newcommand{\bmu}{\mbox{\boldmath${\mu}$}}
\newcommand{\bnabla}{\mbox{\boldmath${\nabla}$}}

\newcommand{\bpsi}{\mbox{\boldmath${\psi}$}}
\newcommand{\brho}{\mbox{\boldmath${\rho}$}}
\newcommand{\bsig}{\mbox{\boldmath${\sigma}$}}
\newcommand{\vecxi}{\mbox{\boldmath${\xi}$}}

\newcommand{\bxi}{\boldsymbol{\xi}}

\newcommand{\bzeta}{\boldsymbol{\zeta}}

\renewcommand{\d}{\mathrm{d}}

\newcommand{\bepsilon}{\mbox{\boldmath${\epsilon}$}}
\newcommand{\bsigma}{\mbox{\boldmath${\sigma}$}}
\renewcommand{\d}{\mathrm{d}}

\newcommand{\vecx}{\boldsymbol{x}}
\newcommand{\vecy}{\boldsymbol{y}}
\newcommand{\vece}{\boldsymbol{e}}

\newcommand{\vecE}{\boldsymbol{E}}

\newcommand{\vecD}{\boldsymbol{D}}
\newcommand{\vecJ}{\boldsymbol{J}}

\begin{document} 
\title{Spectral theory of effective transport for continuous \\uniaxial polycrystalline materials}
\author{N. Benjamin Murphy, Daniel Hallman, Elena Cherkaev, and Kenneth M. Golden}
\address{Department of Mathematics, University of Utah,
155 S 1400 E, Salt Lake City, UT 84112-0090}
\maketitle
\begin{abstract}
Following seminal work in  
the early 1980s that established the existence and representations of the homogenized
transport coefficients for two phase random media,  
we develop a mathematical framework that provides Stieltjes integral representations for the bulk
transport coefficients for uniaxial polycrystalline materials, involving spectral measures of self-adjoint random
operators, which are compositions of non-random and random projection operators. We demonstrate the same
mathematical framework also describes two-component composites, with a simple substitution of the random
projection operator, making the mathematical descriptions of these two distinct physical systems directly
analogous to one another. A detailed analysis establishes the operators arising in each setting are indeed 
self-adjoint on an $L^2$-type Hilbert space, providing a rigorous foundation to the formal spectral theoretic 
framework established by Golden and Papanicolaou in 1983. An abstract extension of the Helmholtz theorem also leads 
to integral representations for the inverses of effective parameters, e.g., effective conductivity and resistivity. An
alternate formulation of the effective parameter problem in terms of a Sobolev-type Hilbert space provides a
rigorous foundation for an approach first established by Bergman and Milton. We show that the correspondence between 
the two formulations is a one-to-one isometry. Rigorous bounds that follow from such 
Stieltjes integrals and partial knowledge about the material geometry are reviewed and validated by numerical 
calculations of the effective parameters for polycrystalline media.
\end{abstract}

\section{Introduction}\label{sec:Introduction}
In this article we formulate a rigorous mathematical framework that provides Stieltjes integral representations for the bulk transport coefficients for  \emph{uniaxial} polycrystalline  \cite{Barabash:JPCM:10323,Gully:PRSA:471:2174,Milton:2002:TC} composite materials, including electrical and thermal conductivity, diffusivity, complex permittivity, and magnetic permeability. All of these transport phenomena are described locally by the same elliptic partial differential equation (PDE)~\cite{Milton:2002:TC}. For example, static electrical conduction \cite{Jackson-1999} is described by $\bnabla \bcdot \, (\bsigma \bnabla \phi)=\rho$ with (free) charge density $\rho$, electrical potential 
$\phi$, electric field $\vecE=-\bnabla\phi$, and local conductivity matrix $\bsigma$ given by \cite{Milton:2002:TC}
$\bsigma=R^T\text{diag}(\sigma_1,\sigma_2,\ldots,\sigma_2)R$ where $\sigma_1$ and $\sigma_2$ are real-valued, i.e., $\bsigma=R^T\text{diag}(\sigma_1,\sigma_2,\sigma_2)R$ for 3D and  $\bsigma=R^T\text{diag}(\sigma_1,\sigma_2)R$ for 2D. This is the general setting for 2D but introduces a uniaxial asymmetry for $d\ge3$, i.e., the local conductivity along one of the crystal axes has the value $\sigma_1$, while the conductivity along all the other crystal axes
has the value $\sigma_2$. Here, $R$ is a rotation matrix determining local crystal orientations.

Consequently, the current and electric fields $\vecJ=\bsigma\vecE$ and $\vecE$ satisfy the electrostatic version of Maxwell's equations $\bnabla\bcdot\vecJ=\rho$ and $\bnabla\btimes\vecE=0$. 
In the long electromagnetic wavelength limit, the electric and displacement fields satisfy the quasistatic limit of Maxwell's equations \cite{Milton:2002:TC} $\bnabla\bcdot\vecD=\rho_b$ and $\bnabla\btimes\vecE=0$, where $\rho_b$ is the bound charge density \cite{Jackson-1999}, the displacement field is given by $\vecD=\bepsilon\vecE$, $\bepsilon=R^T\text{diag}(\epsilon_1,\epsilon_2,\ldots,\epsilon_2)R$, and the crystal permittivities $\epsilon_1(f)$ and $\epsilon_2(f)$ take \emph{complex values} which depend on the electromagnetic wave frequency $f$. In the remainder of this manuscript we formulate the problem of effective transport in terms of electrical conductivity, keeping in mind the broad applicability of the method.

Polycrystalline materials are solids that are composed of many 
crystallites of varying size, shape, and orientation. The crystallites 
are microscopic crystals which are held together by boundaries which can
be highly defective. 
Due to the highly irregular shapes of the crystallites and their
defective boundaries, on a microscopic level, the electromagnetic
transport properties of the medium can be quite erratic. As a
consequence, 
the derivatives in the transport equations may not hold in a classical sense and a weak formulation 
of the transport equation is necessary to provide a rigorous mathematical description of effective transport
for such materials \cite{Papanicolaou:RF-835,Golden:CMP-473}.

In 1982 Papanicolaou and Varadhan \cite{Papanicolaou:RF-835} developed a 
mathematical framework for the homogenization of the transport equation 
$-\bnabla\bcdot(\bsigma\bnabla \phi) = \rho$ for an inhomogeneous composite medium, e.g., a mixture of 
several different materials with different electrical conductivities.  Due to the discontinuous nature of the composite material, 
the derivatives in the transport equation may not hold in a classical sense and a weak formulation 
of the transport equation is necessary to provide a rigorous mathematical description of effective transport
for such materials \cite{Papanicolaou:RF-835}. This mathematical framework holds when 
the components $\sigma_{jk}$, $j,k=1,\ldots,d$, of $\bsigma$ are functions that are spatially periodic, 
almost periodic, or stochastically stationary and ergodic \cite{Papanicolaou:RF-835}. It establishes that the large scale, 
macroscopic electrical transport properties of the system are governed by the homogenized equation
$-\bnabla\bcdot(\bsigma^*\bnabla\bar{\phi}) = \rho$, where $\bsigma^*$ is a (constant) 
\emph{effective conductivity} matrix and $\bar{\phi}$ is a homogenized electrical potential \cite{Bensoussan:Book:1978,Papanicolaou:RF-835}.

In 1983 Golden and Papanicolaou	\cite{Golden:CMP-473} developed the \emph{analytic continuation method},
applying the mathematical framework in \cite{Papanicolaou:RF-835} to the setting of a two-component locally
isotropic composite material. In this setting the components $\sigma_{jk}$, $i,j=1,\ldots,d$, of the local conductivity matrix $\bsigma$ are given by $\sigma_{jk}=\sigma\delta_{jk}$, where $d$ is the system dimension and 
$\delta_{jk}$ is the Kronecker delta. Moreover, $\sigma=\sigma_1\chi_1+\sigma_2\chi_2$, where $\chi_1$ is a characteristic function taking the value 1 in material component one and 0 otherwise, with $\chi_1+\chi_2=1$ and $\chi_i\chi_j=\chi_i\delta_{ij}$. Hence, $\chi_1$ and $\chi_2$ are scalar valued mutually orthogonal  projection operators. In this two-component setting, the components $\sigma^*_{jk}$ of the matrix $\bsigma^*$ have Stieltjes integral representations, introduced by Bergman and Milton  \cite{Bergman:PRL-1285:44,Milton:APL-300,Milton:2002:TC}. 
Golden and Papanicolaou \cite{Golden:CMP-473,Golden:JSP-655:40,Golden:IMA-97} postulated that the Stieltjes measure is a \emph{spectral measure} associated with
the operator $\chi_1\Gamma\chi_1$ \cite{Golden:CMP-473,Murphy:2015:CMS:13:4:825}. Here $\Gamma$ is defined to be $\Gamma=\bnabla(\Delta)^{-1}\bnabla\bcdot$, involving the Laplacian operator $\Delta=\bnabla\bcdot\bnabla=\nabla^2$, and $\chi_1\Gamma\chi_1$ was assumed to be self-adjoint on an appropriate $L^2$-type Hilbert space. The formal analysis by Bergman in  \cite{Bergman:PRC-377:9,Bergman:PRL-1285:44,Bergman:1979:PRB:2359:19:4} is analogous and led to Stieltjes integral representations for the $\sigma^*_{jk}$ in terms of spectral measures associated with an operator on a Sobolev-type Hilbert space. The analytic representation for  the diagonal components $\sigma^*_{kk}$ of $\bsigma^*$ led to rigorous forward and inverse bounds for $\sigma^*_{kk}$ 
\cite{Bergman:PRL-1285:44,Milton:APL-300,Golden:CMP-473,Golden:IMA-97,Milton:2002:TC,Cherkaev:WRM-437,Cherkaev:2001}.

In \cite{Gully:PRSA:471:2174,Barabash:JPCM:10323} the mathematical framework developed in \cite{Golden:CMP-473,Bergman:PRC-377:9,Bergman:1979:PRB:2359:19:4,Milton:APL-300,Milton:2002:TC} was adapted to the setting of uniaxial polycrystalline composite materials, to provide Stieltjes integral representations for the $\sigma^*_{jk}$ in terms of spectral measures associated with operators that were assumed to be self-adjoint on an appropriate Hilbert space. The analysis in \cite{Barabash:JPCM:10323} was 
adapted from the analysis of Bergman in \cite{Bergman:PRC-377:9,Bergman:1979:PRB:2359:19:4}, which involves a Sobolev-type Hilbert space, and \cite{Gully:PRSA:471:2174} adapted the analysis of Golden and Papanicolaou in \cite{Golden:CMP-473} which involves a $L^2$-type Hilbert space. Given the integral representation for the $\sigma^*_{kk}$, both forward and inverse bounds on $\sigma^*_{kk}$ were obtained  \cite{Gully:PRSA:471:2174}.

Here, we write the local conductivity matrix $\bsigma=R^T\text{diag}(\sigma_1,\sigma_2,\ldots,\sigma_2)R$ 
for a uniaxial polycrystalline material  as $\bsigma=\sigma_1X_1+\sigma_2X_2$  
where $X_1$ and $X_2$ are mutually orthogonal projection matrices satisfying $X_1+X_2=I$ and 
$X_i X_j=X_i\delta_{ij}$, where $I$ is the matrix identity. This is a direct analogue of the two-component locally isotropic material setting. We provide a detailed analysis and rigorously establish that 
$\Gamma=\bnabla(\bnabla^*\bnabla)^{-1}\bnabla^*$ is a self-adjoint projection onto the range of $\bnabla$, with respect to an $L^2$-type Hilbert space, where $\bnabla^*$ is the adjoint of a generalized gradient operator $\bnabla$. We prove that 
$X_1\Gamma X_1$ is self-adjoint on the same Hilbert space. The spectral theorem then leads to Stieltjes integral representations for
the components $\sigma^*_{jk}$ of the matrix $\bsigma^*$, involving spectral measures of the self-adjoint operator $X_1\Gamma X_1$. 
These results are extended to provide Stieltjes integral representations for the $\sigma^*_{jk}$ involving spectral measures 
of the self-adjoint operator $X_2\Gamma X_2$.
We provide an abstract formulation of the Helmholtz theorem to extend these results 
to the effective resistivity matrix $\brho^*$, establishing that its components $\rho^*_{jk}$, $j,k=1,\ldots,d$, 
have Stieltjes integral representations involving spectral measures for the operators $X_1\Upsilon X_1$ and $X_2\Upsilon X_2$, where
$\Upsilon=\bnabla\btimes(\bnabla\btimes\bnabla\btimes)^{-1}\bnabla\btimes$ is a self-adjoint projection
onto the range of a generalized curl operator $\bnabla\btimes$.

We develop an alternative formulation that provides Stieltjes integral representations for the components of 
$\bsigma^*$ and $\brho^*$, involving spectral measures for operators that are self-adjoint on a Sobolev-type 
Hilbert space. We show the domains and ranges of the operators associated with the $L^2$-type and Sobolev-type 
Hilbert space formulations are in 1-1 isometric correspondence. We prove that the spectral measures for the two 
formulations are identical, and we provide relationships between the resolution of the identity operators associated 
with each spectral measure.

Due to the direct analogue \cite{Milton:2002:TC}
between the mathematical frameworks for uniaxial polycrystalline media and two-component composite media, 
every result that we establish for the setting of polycrystalline materials also holds for the setting of two-component 
materials. Consequently, we simultaneously provide a rigorous analysis for spectral representations of effective 
parameters for two-component composite materials.

Finally, we review rigorous forward bounds for the effective parameters  \cite{Gully:PRSA:471:2174} that follow 
from their Stieltjes integral representations. We calculate the diagonal components $\epsilon^*_{kk}$ of the 
effective complex permittivity matrix $\bepsilon^*$ for isotropic uniaxial polycrystalline materials with checkerboard 
microstructure and 
with uniform distribution of crystal angles
for both 2D and 3D. We accomplish this by utilizing a discrete, matrix 
formulation of our mathematical framework --- details will be published elsewhere --- and directly calculating the 
spectral measures and effective parameters in terms of eigenvalues and eigenvectors for real-symmetric 
matrix representations for the self-adjoint operators. We demonstrate that numerical calculations for the $\epsilon^*_{kk}$ are captured by rigorous first and second order, nested bounds \cite{Gully:PRSA:471:2174}.

\section{Homogenization and Hilbert space} \label{sec:Homogenizaiton_and_Hilbert}
Towards the goal of formulating a detailed mathematical framework that describes effective transport coefficients 
for uniaxial polycrystalline media, we first review in Section \ref{sec:Homogenizaiton} an abstract theoretical 
framework \cite{Papanicolaou:RF-835,Golden:CMP-473} that addresses the homogenization of the key transport 
equation of interest, e.g., $-\bnabla\bcdot(\bsigma\bnabla \phi) = \rho$ for electrical conductivity,  keeping in mind the broader applicability of the method to thermal conductivity, diffusivity, complex permittivity, and magnetic permeability. In Section \ref{sec:Hilbert_space}, we frame the results of Section \ref{sec:Homogenizaiton} in a Hilbert 
space context which provides analytic properties of the effective transport coefficients and ultimately leads to their 
Stieltjes integral representation involving spectral measures of self-adjoint operators.

\subsection{Homogenization}
\label{sec:Homogenizaiton}
Consider an electrically conductive inhomogeneous medium occupying a $d$-dimensional region of space $\Vc\subset\mathbb{R}^d$, e.g., a composite or mixture of several different materials with different conductivities. 
A key assumption of homogenization theory \cite{Bensoussan:Book:1978,Papanicolaou:RF-835,Golden:CMP-473} is that the value of the local conductivity matrix $\bsigma(\vecx)$ changes 
rapidly as $\vecx$ varies over lengths comparable to the size of the region $\Vc$. This is 
modeled by introducing a small scalar $\delta>0$ and expressing the transport equation as
\begin{align}\label{eq:elliptic_delta}
-\bnabla\bcdot\left(
\bsigma\left(\frac{\vecx}{\delta}\right)
\bnabla\phi^\delta(\vecx)
\right) = \rho(\vecx)\,,
\quad
\vecx\in\Vc,
\end{align}
where differentiation is over the ``slow variable'' $\vecx$ and we define the ``fast variable'' $\vecy=\vecx/\delta$. 
The mathematical framework holds for general homogeneous boundary conditions \cite{Bensoussan:Book:1978,Papanicolaou:RF-835}.

Expanding $\phi^\delta$ in powers of $\delta$ and equating coefficients 
with like powers of $\delta$ shows, in the limit $\delta\to0$, that the large 
scale, macroscopic electrical transport of the system is described by an 
equation similar to \eqref{eq:elliptic_delta} but the inhomogeneous medium is replaced by a homogeneous medium with a \emph{constant} 
effective electrical conductivity matrix $\bsigma^*$ and a homogenized electrical potential $\bar{\phi}$ satisfying the same boundary conditions as $\phi^\delta$ \cite{Bensoussan:Book:1978,Papanicolaou:RF-835}
\begin{align}\label{eq:elliptic_effective}
-\bnabla\bcdot(
\bsigma^*
\bnabla\bar{\phi}(\vecx)
) = \rho(\vecx)\,,
\quad
\vecx\in\Vc.
\end{align}
%
The effective conductivity matrix $\bsigma^*$ is given by \cite{Bensoussan:Book:1978,Papanicolaou:RF-835}  
\begin{align}
\label{eq:effective_conductivity_matrix}
\bsigma^* = \left\langle
\bsigma(\vecy)\left(I + \bnabla\bpsi(\vecy)\right)
\right\rangle,
%
\end{align}
where $I$ is the identity matrix on $\mathbb{R}^d$, $\langle\cdot\rangle$ denotes volume average in the fast 
variable $\vecy=\vecx/\delta$, and partial differentiation associated with the gradient operator $\bnabla$ in 
\eqref{eq:effective_conductivity_matrix} is with respect to the fast variable $\vecy$. The components $\psi_i$ of the 
row vector $\bpsi=(\psi_i,\ldots,\psi_d)$ satisfy the following \emph{cell problem}  
\begin{align}\label{eq:cell_problem}
-\bnabla\bcdot( \bsigma(\vecy)
(\vece_i+\bnabla\psi_i(\vecy)))=0\,,
\quad
\langle\bnabla\psi_i\rangle=0\,,
\quad 
i=1,\ldots,d,
\end{align}
which involves only the fast variable $\vecy$, where $\vece_i$ is the $i$th direction standard basis vector 
\cite{Bensoussan:Book:1978,Papanicolaou:RF-835}.

The physical idea behind equations \eqref{eq:effective_conductivity_matrix} and \eqref{eq:cell_problem} is as 
follows. Define (in the fast variable $\vecy$) the local ``electric field'' $\vecE=\vece_i+\vecE_f$ with  mean zero
fluctuating field $\vecE_f=\bnabla \psi_i$ and $\langle\vecE\rangle=\vece_i$, and define the associated local
``electrical current density'' $\vecJ=\bsigma\vecE$. Then, equation \eqref{eq:cell_problem} can be written as the
electrostatic form of Maxwell's equations $\bnabla\bcdot\vecJ=0$ with $\bnabla\btimes\vecE=0$, with no source term $\rho$, where the vector
valued functions and differentiation depend only on the variable $\vecy$. Applying $\vece_i$ to both sides of
equation \eqref{eq:effective_conductivity_matrix} shows the average of $\vecJ$ can be written as
\begin{align}\label{eq:avg_J}
\langle\vecJ\rangle=\bsigma^*\langle\vecE\rangle.
\end{align}
Therefore, the local conductivity in $\vecJ=\bsigma\vecE$ relates the local ``electric field'' to the local ``current 
density'' while the effective conductivity $\bsigma^*$ relates the average ``electric field'' to the average ``current 
density.''

To summarize, equations \eqref{eq:elliptic_delta}--\eqref{eq:cell_problem} define a \emph{homogeneous} medium
with constant conductivity $\bsigma^*$ that behaves \emph{macroscopically} the same as the inhomogeneous 
medium with spatially dependent conductivity $\bsigma(\vecx)$. In particular, if $\bar{\phi}$ satisfies equation
\eqref{eq:elliptic_effective} with $\bsigma^*$ given by \eqref{eq:effective_conductivity_matrix}, then 
$\lim_{\delta\downarrow0}\phi^\delta=\bar{\phi}$ in a sense that depends on the nature of the composite medium in
$\Vc$ determined by $\bsigma(\vecx)$, which can be periodic, almost periodic, or stochastically stationary
\cite{Papanicolaou:RF-835}. For the stationary and almost periodic settings, such limit theorems and existence proofs
for equations \eqref{eq:elliptic_effective}--\eqref{eq:cell_problem} are given in terms of infinite volume limits
$\Vc\to\mathbb{R}^d$ or directly using $\Vc=\mathbb{R}^d$ \cite{Bensoussan:Book:1978,Papanicolaou:RF-835}. To set ideas, we will
assume the stationary stochastic setting with $\Vc=\mathbb{R}^d$.

\subsection{Weak solutions to Maxwell's equations}
\label{sec:Hilbert_space}
We start this section with a formal discussion to motivate the abstract Hilbert space approach utilized later in the 
section. Fix a geometric realization of a random polycrystalline medium and consider the problem of solving Maxwell's
equations for electrostatics for the geometric realization
\begin{align}   \label{eq:Maxwells_Equations_E0}  
&\bnabla \times\vecE =0, 
\quad
\bnabla \bcdot\vecJ=0,
\quad
\vecJ=\bsigma\vecE. 
\end{align}
Here, $\vecE $ and $\vecJ$ are the physical electric field and current density within
the polycrystalline medium, respectively. 
If the components of the vector field $\vecE $ were continuously differentiable 
on all of $\mathbb{R}^d$ then  
$\bnabla \times\vecE =0$ would imply there exists a scalar 
potential $\varphi$ such that
$\vecE = \vecxi + \bnabla \varphi$ with arbitrary constant vector $\vecxi$ \cite{Jackson-1999}. 
In this case, equation \eqref{eq:Maxwells_Equations_E0} is equivalent to
\begin{align}\label{eq:Alternate_Maxwells_E}  
\bnabla \bcdot\bsigma(\vecxi + \bnabla \varphi)=0\,.
\end{align}
%
However, the functions $\sigma_{jk}$, $j,k=1,\ldots,d$, can be discontinuous, hence non-differentiable, 
across the crystallite boundaries of the polycrystalline medium, so equation
\eqref{eq:Alternate_Maxwells_E} may not be defined in a classical sense. A
weak form~\cite{Folland:95:PDEs} of equation \eqref{eq:Alternate_Maxwells_E} is given by 
\cite{Bensoussan:Book:1978}
%
\begin{align}\label{eq:Weak_Alternate_Maxwell_E}
\langle\bsigma(\vecxi + \bnabla \varphi)\bcdot\bnabla \psi\rangle=0,
\end{align}
where $\psi$ is an infinitely differentiable function with compact support, and we
emphasize that $\bnabla \psi$ is a curl-free vector field. This would directly bypass the 
non-differentiability of the $\sigma_{jk}$. Although, indirectly, the discontinuous nature of these 
functions might still render the potential $\varphi$  non-differentiable. Consequently, the weak form 
\eqref{eq:Weak_Alternate_Maxwell_E}, again, may not be defined in a classical
sense.  We address such issues by providing abstract Hilbert space
formulation~\cite{Papanicolaou:RF-835,Golden:CMP-473} of 
equation~\eqref{eq:Weak_Alternate_Maxwell_E}.

The abstract framework presented in \cite{Papanicolaou:RF-835} unifies the homogenization problem discussed in
Section \ref{sec:Homogenizaiton} for settings where $\bsigma$ is periodic, almost periodic, and random with just
stationarity, ergodicity, and uniform ellipticity assumptions. Here, we only review details \cite{Papanicolaou:RF-835,Golden:CMP-473} needed to provide Stieltjes integral representations for the effective
transport coefficients for uniaxial polycrystalline media for the stochastically stationary setting, as the other settings
directly follow from this formulation \cite{Papanicolaou:RF-835}. Consider the random setting and let $(\Omega,P)$ 
be a probability space with underlying sigma-algebra of $P$-measurable sets and $\omega\in\Omega$ labeling the 
geometric realization of the random medium with electrical conductivity $\bsigma=\bsigma(\vecx,\omega)$ with
$\vecx\in\mathbb{R}^d$. Consider the Hilbert space $\Hs_0=L^2(\Omega,P)$ of scalar valued functions. (Later we
will focus on the Hilbert space $\Hs=\bigotimes_{i=1}^d \Hs_0$ of $d$-dimensional vector valued functions).
Furthermore, let $(\mathbb{R}^d,\nu)$ be a measure space with underlying sigma-algebra of Lebesgue measurable 
cylinder sets with rational diameters and $\nu$ denotes the Lebesgue measure. Consider the Hilbert space 
$L^2(\mathbb{R}^d,\nu)$~\cite{Papanicolaou:RF-835}.

We restrict the support of $P$ to a uniformly bounded, strictly stationary 
stochastic process $\bsigma(\vecx,\omega)$ satisfying
\begin{align}\label{eq:bounded_sigma}
\sigma_0|\bxi|^2\le
\bsig(\vecx,\omega)\bxi\bcdot\bxi
\le\sigma_0^{-1} |\bxi|^2
\end{align}
for some constant $\sigma_0>0$ and for all $\vecx,\bxi\in\mathbb{R}^d$ 
and $\omega\in\Omega$. Strict stationarity means for any $\vecxi\in\mathbb{R}^d$, 
any integer $n=1,2,\ldots,$ and any points $\vecx_1,\ldots,\vecx_n$ in $\mathbb{R}^d$, 
that the joint distribution of 
$\sigma_{jk}(\vecx _1,\omega),\ldots,\sigma_{jk}(\vecx _n,\omega)$ 
and that of $\sigma_{jk}(\vecx_1+\vecxi,\omega),\ldots,\sigma_{jk}(\vecx_n+\vecxi,\omega)$ 
are the same for all $j,k=1,\ldots,d$~\cite{Golden:CMP-473,Papanicolaou:RF-835}. We also 
assume that the processes $\sigma_{jk}(\vecx,\omega)$ are stochastically 
continuous~\cite{Papanicolaou:RF-835}. We assume that the probability measure $P$ is 
invariant with respect to the translation group 
$\tau_x:\Omega\to\Omega$~\cite{Papanicolaou:RF-835}. More specifically, we assume 
that the group is one-to-one and preserves the measure $P$, i.e. 
$P(\tau_x A)=P(A)$ for all $P$-measurable sets $A$, with $\tau_x\tau_y=\tau_{x+y}$
\cite{Golden:CMP-473,Papanicolaou:RF-835}.
The group of transformations $\tau_x$ on $\Omega$ induces a group of
operators $T_x$ on the Hilbert space $\Hs_0$ defined by
$(T_x f^\prime)(\omega)=f^\prime(\tau_{-x}\,\omega)$ for all $f^\prime\in \Hs_0$. 
We may associate the stationary process $f$ with 
$f(\vecx,\omega)=(T_x f^\prime)(\omega)=f^\prime(\tau_{-x}\,\omega)$, 
for all $f^\prime\in \Hs_0$. The prime notation 
is used to associate a function with its translates that form the stationary process.
We assume as in \cite{Papanicolaou:RF-835} that there are measurable functions $\sigma^{\,\prime}_{jk}(\omega)=\sigma_{jk}(0,\omega)$ on $\Omega$ such that
%
\begin{align}\label{eq:Stationary_sig}
\sigma_{jk}(\vecx ,\omega)=
\sigma^{\,\prime}_{jk}(\tau_{-x}\,\omega), 
\quad
\forall \ \vecx \in\mathbb{R}^d, \ \omega\in\Omega. 
\end{align}

Since the group generated by $\tau_x$ (through composition) is measure preserving, 
the operators $T_x$ form a unitary group. When a single coordinate variable $x_i$ 
of $\vecx=(x_1,\ldots,x_d)$ varies at a time in the group $T_x$, we obtain 
$d$ one-parameter strongly continuous unitary groups in $\Hs_0$ that 
commute with each other. Let $L_1,L_2,\ldots,L_d$ denote the infinitesimal 
generators of these groups, which are closed and densely defined with domains 
$\mathscr{D}_i\subset \Hs_0$~\cite{Papanicolaou:RF-835}.  They commute with one another 
$L_iL_j=L_jL_i$ and they are skew adjoint with respect to the sesquilinear 
$\Hs_0$-inner-product $\langle\cdot,\cdot\rangle_0$, satisfying 
$\langle L_i f^\prime, g^\prime\rangle_0=-\langle f^\prime,L_i g^\prime\rangle_0$
for all $f^\prime, g^\prime\in\mathscr{D}_i$~\cite{Golden:CMP-473,Papanicolaou:RF-835}. 
For $f^\prime\in\mathscr{D}_i$ we have  
\begin{align}\label{eq:Li}
(L_i f^\prime)(\omega)=
\left.\frac{\partial}{\partial x_i}(T_x f^\prime)(\omega) \right|_{x=0} 
\end{align}
where differentiation in~\eqref{eq:Li} is defined in the sense of 
convergence in $\Hs_0$ for elements of $\mathscr{D}_i$
\cite{Golden:CMP-473,Papanicolaou:RF-835}.

The closed linear space $\mathscr{D}=\cap_{i=1}^d\mathscr{D}_i$ of $\Hs_0$ becomes 
a Hilbert space with inner product $\langle\cdot,\cdot\rangle_{0,1}$ 
given by 
$\langle f^\prime,g^\prime\rangle_{0,1}=\langle f^\prime,g^\prime\rangle_0+\sum_{i=1}^d\langle
L_if^\prime,L_ig^\prime\rangle_0$~\cite{Golden:CMP-473,Papanicolaou:RF-835}.
The Hilbert space can be identified with the space of all stationary random 
processes on $\mathbb{R}^d$, $\Hs_S(\mathbb{R}^d;\Hs_0^1)$~\cite{Papanicolaou:RF-835}. Similarly,
the Hilbert space $\Hs_0^1$ can be identified with the set of mean square
differentiable, stationary random processes 
$\Hs^1_S(\mathbb{R}^d;\Hs_0^1)$~\cite{Papanicolaou:RF-835}.
If $f^\prime\in\Hs_0^1$ then its spatial partial derivatives form stationary processes and 
\begin{align}\label{eq:weak_derivative}
\frac{\partial f(\vecx,\omega)}{\partial x_i}=
L_i f(\vecx,\omega),
\quad
\nu\times P 
\ \
a.e.
\end{align}
with equality holding $\nu\times P$ almost everywhere~\cite{Papanicolaou:RF-835}.

In this abstract setting the curl and divergence formulas,
$\bnabla\times\vecE=0$ and $\bnabla\bcdot\vecJ=0$ are interpreted in the 
following weak sense~\cite{Golden:CMP-473}, respectively
\begin{align}\label{eq:Li_curl_div}
L_i E^\prime_j - L_j E^\prime_i = 0 
\text{ for all } 
i,j=1,\ldots,d, 
\text{  and  } \sum_{i=1}^d L_i J^\prime_i = 0\,,
\end{align}
involving the generalized gradient operator $\bnabla=(L_1,\ldots,L_d)$.
Consider the Hilbert space $\Hs=\bigotimes_{i=1}^d \Hs_0$ with
\emph{sesquilinear} inner product $\langle\cdot,\cdot\rangle$ defined by
$\langle\vecxi^\prime,\bzeta^\prime \rangle=
\langle\vecxi^\prime\bcdot\bzeta^\prime \rangle=
\sum_i\langle\xi_i^\prime,\zeta_i^\prime\rangle_0$, 
with
$\langle\bzeta^\prime ,\vecxi^\prime\rangle=
\overline{\langle\vecxi^\prime,\bzeta^\prime \rangle}$, 
where $\xi_i^\prime$ is the $i$th component of the vector $\vecxi^\prime$, 
$\langle\cdot\rangle$ denotes ensemble average over $\Omega$ with respect to the
measure $P$,
and 
$\overline{z}$
denotes complex conjugation of the complex number $z$. 
Now define the
Hilbert spaces  of ``curl-free'' and ``divergence-free'' random fields~\cite{Golden:CMP-473}   
\begin{align}\label{eq:curlfreeHilbert}
&\Hs_\times=
\left\{\vecxi^\prime(\omega)\in \Hs \ | \ L_i\xi_j^{\,\prime}-L_j\xi_i^{\,\prime}=0 
\ \text{ weakly and }
\langle\vecxi^\prime\,\rangle=0
\right\}, \\
&\Hs_{\bullet}=
\left\{\bzeta ^\prime(\omega)\in \Hs \ \Big| \ \sum_{i=1}^dL_i\zeta_i^{\,\prime}=0 
\ \text{ weakly and }
\langle\bzeta^\prime\,\rangle=0\right\}\,,
\notag 
\end{align}  
where we use the simplified vector notation 
$\langle\vecxi^\prime\,\rangle=
(\langle\xi_1^\prime\rangle,\ldots,\langle\xi_d^\prime\rangle)^T$ and the completion of the spaces 
$\Hs_\times$ and $\Hs_\bullet$ are with respect to the norm induced by the $\Hs$-inner-product
$\|\cdot\|^2=\langle\cdot,\cdot\rangle$.

We are now ready to state the main results of this section. Consider the problem of solving Maxwell's equations for 
electrostatics in this abstract setting: find stationary random vector fields 
$\vecE (\vecx ,\omega)$ and $\vecJ(\vecx ,\omega)$
such that
\begin{align}   \label{eq:Maxwells_Equations_E}  
&\bnabla \times\vecE =0, \quad
\bnabla \bcdot\vecJ=0,\quad
\vecJ=\bsigma\vecE ,\quad
\langle\vecE \,\rangle=\vecE _0\,. 
\end{align}
Here,
$\vecE $ and $\vecJ$ are the physical electric field and current density within
the polycrystalline medium, respectively. 
The average $\vecE _0$ is assumed to be \emph{given}, and by stationarity it is independent of $\vecx \in\mathbb{R}^d$.  More specifically, we seek a stationary solution $\vecE $ to equation
\eqref{eq:Maxwells_Equations_E} of the form    
%
\begin{align}\label{eq:Stationary_E}
\vecE (\vecx ,\omega)=\vecE ^{\,\prime}(\tau_{-x}\,\omega), \quad
\forall \ \vecx \in\mathbb{R}^d, \ \omega\in\Omega,
\end{align}
where $\vecE ^{\,\prime}(\omega)=\vecE (0,\omega)$ and
$\vecE _0=\langle\vecE ^{\,\prime}(\omega)\rangle$. We are ready to state one of the main 
results of this section in the following theorem. 

\begin{theorem}
\label{thm:existenceE}
If $\bsigma$ satisfies \eqref{eq:bounded_sigma} then there exists 
unique $\vecE_f\in\Hs_{\times}$ satisfying $\langle\vecE_f\rangle=0$ such that 
$\vecE (\vecx ,\omega)=\vecE _0+\vecE _f(\vecx ,\omega)$ and $\vecE$ satisfies
the systems of partial differential equations in \eqref{eq:Maxwells_Equations_E}
in the weak sense discussed above. 
\end{theorem}

\begin{proof}
In view of equations~\eqref{eq:Weak_Alternate_Maxwell_E},~\eqref{eq:Stationary_sig},
and \eqref{eq:Stationary_E}, we consider the following variational problem
\cite{Golden:CMP-473}. Given $\bsigma^{\,\prime}$ and $\vecE _0$ find 
$\vecE _f^{\,\prime}\in\Hs_\times$ 
such that         
\begin{align}
    \label{eq:Weak_Curl_Free_Variational_Form_E}
    \langle\bsigma^{\,\prime}(\omega)
    (\vecE _0+\vecE _f^{\prime}(\omega))
    \bcdot\vecxi^{\prime}(\omega)\rangle=0
    \quad  \forall \ \vecxi^{\prime}\in\Hs_\times\,.
\end{align}
We rewrite equation~\eqref{eq:Weak_Curl_Free_Variational_Form_E}  as \cite{Golden:CMP-473}
%
\begin{align}
    \label{eq:Bilinear_functional_E}
    \Phi(\vecE _f^{\,\prime},\vecxi^{\,\prime})=f_\sigma(\vecxi^{\,\prime}),
    \qquad  
    \Phi(\vecE _f^{\,\prime},\vecxi^{\,\prime})
    =\langle\bsigma^{\,\prime}(\omega)\vecE _f^{\,\prime}(\omega)\bcdot\vecxi^{\,\prime}(\omega)\rangle,
    \quad
    f_\sigma(\vecxi^{\,\prime})=
    -\langle\bsigma^{\,\prime}(\omega)\vecE _0\bcdot\vecxi^{\,\prime}(\omega)\rangle
\end{align}
involving the bilinear functional $\Phi$ and the linear functional $f_\sigma$.

By the Cauchy--Schwartz inequality and equation \eqref{eq:bounded_sigma}, 
$\Phi$ and $f_\sigma$ are bounded \cite{Folland:99:RealAnalysis,Folland:95:PDEs,Golden:CMP-473}.  
For now, we will assume that $\Phi$ is 
coercive, i.e., that there exists a positive constant 
$\kappa>0$ such that $\Phi(\vecxi,\vecxi\,)\geq\kappa\|\vecxi\,\|^2$ for all 
$\vecxi\in\Hs_\times$, where $\|\cdot\|$ denotes the norm induced by the
$\Hs$-inner-product \cite{Folland:99:RealAnalysis}. 
Later, we will demonstrate that the
coercivity condition determines analytic properties of the
effective conductivity matrix $\bsigma^*$ in \eqref{eq:avg_J}. 
By the Lax-Milgram theorem \cite{Folland:95:PDEs}, there exists unique
$\vecE _f^{\,\prime}\in\Hs_\times$ 
that satisfies equation~\eqref{eq:Bilinear_functional_E}  
for all $\vecxi^{\,\prime}\in\Hs_\times$. 
\end{proof}

We now provide an alternate approach to solving Maxwell's Equations 
in terms of the resistivity matrix $\brho=\bsigma^{-1}$ 
(the matrix inverse of the conductivity matrix). 
Consider the alternate problem 
\begin{align}   \label{eq:Maxwells_Equations_J}  
\bnabla \times\vecE =0, \quad
\bnabla \bcdot\vecJ=0,\quad
\vecE=\brho\vecJ ,\quad
\langle\vecJ \,\rangle=\vecJ _0\,. 
\end{align}
More specifically, we seek a stationary solution $\vecJ$ to equation
\eqref{eq:Maxwells_Equations_J} of the form    
\begin{align}\label{eq:Stationary_J}
\vecJ (\vecx ,\omega)=\vecJ ^{\,\prime}(\tau_{-x}\,\omega), \quad
\forall \ \vecx \in\mathbb{R}^d, \ \omega\in\Omega,
\end{align}
where $\vecJ ^{\,\prime}(\omega)=\vecJ (0,\omega)$ and
$\vecJ _0=\langle\vecJ ^{\,\prime}(\omega)\rangle$. 
We assume as in \cite{Papanicolaou:RF-835,Golden:CMP-473} that there are measurable functions $\rho^{\,\prime}_{jk}(\omega)=\rho_{jk}(0,\omega)$ on $\Omega$ such that
%
\begin{align}\label{eq:Stationary_rho}
\rho_{jk}(\vecx ,\omega)=
\rho^{\,\prime}_{jk}(\tau_{-x}\,\omega), 
\quad
\forall \ \vecx \in\mathbb{R}^d, \ \omega\in\Omega. 
\end{align}
We restrict the support of $P$ to a uniformly bounded, strictly stationary 
stochastic process $\brho(\vecx,\omega)$ satisfying
\begin{align}\label{eq:bounded_rho}
\rho_0|\bxi|^2\le
\brho(\vecx,\omega)\bxi\bcdot\bxi
\le\rho_0^{-1} |\bxi|^2
\end{align}
for some constant $\rho_0>0$ and for all $\vecx,\bxi\in\mathbb{R}^d$ 
and $\omega\in\Omega$.

\begin{corollary}
If $\brho$ satisfies \eqref{eq:bounded_rho}  then there exists 
unique $\vecJ_f\in\Hs_{\bullet}$ satisfying $\langle\vecJ_f\rangle=0$ such that 
$\vecJ (\vecx ,\omega)=\vecJ _0+\vecJ _f(\vecx ,\omega)$ and $\vecJ$ satisfies
the systems of partial differential equations in \eqref{eq:Maxwells_Equations_J}
in the weak sense discussed above. 
\label{cor:existenceJ}
\end{corollary}

The proof of Corollary  \ref{cor:existenceJ} is analogous to that of Theorem  
\ref{thm:existenceE}. By construction,
\begin{align}\label{eq:Unique_Solutions}
\vecE ^{\,\prime}(\omega)&=\vecE _0+\vecE _f^{\,\prime}(\omega), \qquad
\vecJ^{\,\prime}(\omega)=\bsigma^{\,\prime}(\omega)\vecE ^{\,\prime}(\omega),
\\
\vecJ^{\,\prime}(\omega)&=\vecJ_0+\vecJ_f^{\,\prime}(\omega), \qquad
\vecE ^{\,\prime}(\omega)=\brho^{\,\prime}(\omega)\vecJ^{\,\prime}(\omega),
\end{align}
are the unique (weak) solutions of equations~\eqref{eq:Maxwells_Equations_E}
and~\eqref{eq:Maxwells_Equations_J}, respectively, via equations
\eqref{eq:Stationary_sig}, \eqref{eq:Stationary_E},  
\eqref{eq:Stationary_rho}, \eqref{eq:Stationary_J}. 
To simplify notation, we will henceforth drop
the distinction between the primed variables $\vecE _f^{\,\prime}(\omega)$ and 
$\vecE _f(\vecx ,\omega)$, for example, as the context of each notation
is now clear. Note that since
$\vecE _f\in\Hs_\times$ and $\vecJ_f\in\Hs_\bullet$, by equation
\eqref{eq:Weak_Curl_Free_Variational_Form_E}, 
$\langle\bsigma(\vecE_0+\vecE_f)\bcdot\vecxi\rangle=0$ for all $\vecxi\in\Hs_\times$, 
and its analogue, 
$\langle\brho(\vecJ_0+\vecJ_f)\bcdot\bzeta\rangle=0$
for all $ \bzeta\in\Hs_{\bullet}$,
we have the energy \cite{Jackson-1999} constraints $\langle\vecJ\bcdot\vecE _f\rangle=0$ and
$\langle\vecE \bcdot\vecJ_f\rangle=0$, respectively, which lead to the following
reduced energy representations $\langle\vecJ\bcdot\vecE \,\rangle=\langle\vecJ\,\rangle\bcdot\vecE _0$
and $\langle\vecE \bcdot\vecJ\,\rangle=\langle\vecE \,\rangle\bcdot\vecJ_0$. Consequently, by equation
\eqref{eq:avg_J}, $\langle\vecJ\rangle=\bsigma^*\langle\vecE\rangle$, and its analogue 
$\langle\vecE\rangle=\brho^*\langle\vecJ\rangle$ for the alternative problem, 
we have the following energy representations
involving the effective parameters  
%
\begin{align}\label{eq:Energy_Reps}
\langle\vecJ\bcdot\vecE \rangle=\bsigma^*\vecE _0\bcdot\vecE _0=\brho^*\vecJ_0\bcdot\vecJ_0.
\end{align}

Before we move on to the next section to discuss the analytic properties of the effective parameters $\bsigma^*$ and $\brho^*$, we briefly discuss the distinction between the \emph{definition} of the effective conductivity
matrix given in terms of the energy equation \eqref{eq:Energy_Reps} and the effective conductivity matrix arising from the homogenization theorem. The homogenization theorem leads to the formulas for the effective conductivity matrix in equations \eqref{eq:effective_conductivity_matrix} and \eqref{eq:avg_J} which involves the solution to 
the cell problem in \eqref{eq:cell_problem}. Specifically, the effective conductivity arising in the homogenization
theorem involves the \emph{non-physical} ``electric field" $\vecE=\vece_i+\bnabla\psi_i$ and ``current density"
$\vecJ=\bsigma\vecE$  and yields the formula for the effective conductivity 
$\langle\vecJ\rangle=\bsigma^*\langle\vecE\rangle$. This connects the effective conductivity to the homogenized 
equation in \eqref{eq:elliptic_effective}. Alternatively, we may start with the electrostatic version (or quasistatic limit of) Maxwell's equations involving the \emph{physical} electric field $\vecE$ and current density $\vecJ$ (or displacement field $\vecD$), without a source term, and bypass the discussion regarding slow and fast 
variables, and use the results given in the current section to simply \emph{define} the effective conductivity in 
terms of system energy as in \eqref{eq:Energy_Reps}. This alternative approach, while 
lacking the connection to the homogenized equation, instead provides an analytic representation for the physical 
energy, which is beneficial for statistical mechanics models of composite materials --- details will be 
published elsewhere. For the remainder of this manuscript we will focus on the alternative approach involving
physical fields and take $\langle\vecJ\rangle=\bsigma^*\langle\vecE\rangle$ and $\langle\vecE\rangle=\brho^*\langle\vecJ\rangle$ as a definitions.

\section{Analytic properties of effective parameters}
\label{sec:Analytic_properties}
The formulation of the effective parameter problem in Section \ref{sec:Homogenizaiton_and_Hilbert} holds 
for quite general periodic, almost periodic, and random composite materials 
\cite{Papanicolaou:RF-835,Golden:CMP-473}. For example, multi-component
composite materials have locally isotropic microgeometry with conductivity function that varies from place to place  
in space, for each $\omega\in\Omega$. Consequently, $\sigma_{jk}(\vecx,\omega)=\sigma(\vecx,\omega)\delta_{jk}$, 
where $\delta_{jk}$ is the Kronecker-delta, and similarly $\rho_{jk}(\vecx,\omega)=\rho(\vecx,\omega)\delta_{jk}$.
For $n$-components, it is clear that the scalar functions
$\sigma(\vecx,\omega)$ and $\rho(\vecx,\omega)$ can be written as 
\begin{align}\label{eq:projection_decomp_sigma_multicomp}
\sigma(\vecx,\omega)=\sum_{i=1}^n\sigma_i\chi_i(\vecx,\omega),  
\quad
\rho(\vecx,\omega)=\sum_{i=1}^n\rho_i\chi_i(\vecx,\omega),  
\qquad
\chi_i\chi_j=\chi_i\delta_{ij}\,,
\quad
\sum_{i=1}^n\chi_i=1\,,
\end{align}
where $\rho_i=1/\sigma_i$ and the characteristic function, $\chi_i$, takes the value $\chi_i(\vecx,\omega)=1$ 
when the material has component $i$ located at $\vecx$ for $\omega\in\Omega$ and 
$\chi_i(\vecx,\omega)=0$ otherwise. 

For two-component materials ($n=2$), Stieltjes integral representations can be obtained for the components $\sigma^*_{jk}$, $j,k=1,\ldots,d$, of the effective conductivity matrix $\bsigma^*$,  
involving spectral measures of a self-adjoint random operator and a single complex variable \cite{Bergman:AP-78,Golden:CMP-473,Milton:APL-300,Murphy:2015:CMS:13:4:825}.  Rigorous bounds for the diagonal components $\sigma^*_{kk}$ can be obtained \cite{Bergman:AP-78,Milton:APL-300} using the theory of Pad\'{e} approximants associated with such Stieltjes integrals \cite{Baker:1996:Book:Pade}, as well as the theory of convex combinations of extremal measures \cite{Bergman:AP-78,Golden:CMP-473}. 
When $n\ge3$, rigorous bounds for $\sigma^*_{kk}$ can be obtained using methods of complex analysis for functions of several ($n-1$) complex variables, using a polydisk integral representation of $\sigma^*_{kk}$ \cite{Golden:JSP-655:40}.

\subsection{General polycrystalline media}\label{sec:general_polycrystalline_media}
We now show that the theory of effective transport for multicomponent media and that for general polycrystalline
media are direct analogues of one another. This follows from the fact that the conductivity matrix $\bsigma$ for polycrystalline media can be written as a direct analogue of equation 
\eqref{eq:projection_decomp_sigma_multicomp}.
\begin{lemma}\label{lem:projection_decomp_sigma_poly}
The conductivity $\bsigma$ and resistivity $\brho$ matrices for general polycrystalline media can be written as
\begin{align}\label{eq:projection_decomp_sigma_poly}
    \bsigma(\vecx,\omega)=\sum_{i=1}^d\sigma_i X_i(\vecx,\omega),  
    \quad
    \brho(\vecx,\omega)=\sum_{i=1}^d\rho_i X_i(\vecx,\omega),  
    \qquad
    X_i X_j=X_i\delta_{ij}\,,
    \quad
    \sum_{i=1}^d X_i=I\,,
\end{align}
where $I$ is the identity matrix on $\mathbb{R}^d$ and the $X_i$, $i=1,\ldots,d$, are mutually orthogonal
projection matrices.
\end{lemma}
\begin{proof}
The conductivity matrix $\bsigma$ for  polycrystalline media is given by   
\begin{align}\label{eq:polycrystal_parameters}
    \bsigma(\vecx ,\omega)=R^{\,T}(\vecx ,\omega)\text{diag}(\sigma_1,\sigma_2,\ldots,\sigma_d)R(\vecx ,\omega),
\end{align}
where $R(\vecx ,\omega)$ is a rotation matrix satisfying
$R^{\,T}=R^{\,-1}$, where $R^{\,T}$ and $R^{\,-1}$ denote transpose and inverse of the matrix $R$. 
For example, when $d=2$ we have 
\begin{align}\label{eq:polycrystal_parameters_2D}
    \bsigma=R^{\,T}
    \left[
    \begin{array}{cc}
        \sigma_1& 0\\
        0 & \sigma_2\\
    \end{array}
    \right]
    R\,,
    \quad
    R=
    \left[
    \begin{array}{rr}
        \cos\theta& -\sin\theta\\
        \sin\theta & \cos\theta\\
    \end{array}
    \right]\,,
\end{align}
where $\theta=\theta(\vecx ,\omega)$ is the orientation angle of the crystal which has an 
interior containing $\vecx \in\mathbb{R}^d$ for $\omega\in\Omega$, measured from the
direction $\vece _1$. Here, $\vece _j$,
$j=1,\ldots,d$, are standard basis vectors with components
$(\vece _j)_k=\delta_{jk}$. In higher
dimensions, $d\geq3$, the rotation matrix $R$ is a composition of
``basic'' rotation matrices 
$R_i$, e.g., $R=\prod_{j=1}^dR_j$, where the matrix $R_j(\vecx ,\omega)$
rotates vectors in $\mathbb{R}^d$ by an angle
$\theta_j=\theta_j(\vecx ,\omega)$ about the $\vece _j$ axis. For example, in three
dimensions 
\begin{align}\label{eq:polycrystal_parameters_3D}
    R_1=
    \left[
    \begin{array}{ccc}
        1     &0     &    0\\
        0     &\cos\theta_1 & -\sin\theta_1\\
        0     &\sin\theta_1 & \cos\theta_1\\
    \end{array}
    \right],
    \quad
    R_2=
    \left[
    \begin{array}{ccc}
        \cos\theta_2  & 0     &\sin\theta_2 \\
        0      & 1     & 0\\
        -\sin\theta_2 & 0     &\cos\theta_2\\
    \end{array}
    \right],
    \quad
    R_3=
    \left[
    \begin{array}{ccc}
        \cos\theta_3 & -\sin\theta_3 & 0\\
        \sin\theta_3 & \cos\theta_3  & 0\\
        0     & 0      & 1\\
    \end{array}
    \right].
\end{align}

Utilizing the projection matrices $C_i=\text{diag}(\vece_i)$, $i=1,\ldots,d$, 
equation~\eqref{eq:polycrystal_parameters} can be written as
\begin{align}
    \bsigma(\vecx,\omega)=\sum_{i=1}^n\sigma_i X_i(\vecx,\omega),
    \quad
    X_i=R^T(\vecx,\omega)C_i   R(\vecx,\omega).	  
\end{align}
The formula for $\bsigma$ in equation \eqref{eq:projection_decomp_sigma_poly} then follows from the fact that $R$ is a rotation matrix satisfying $R^T=R^{-1}$ and the $C_i$ are mutually orthogonal projection matrices satisfying $C_i C_j=C_i\delta_{ij}$ and $\sum_i C_i=I$. The formula for $\brho$ in equation \eqref{eq:projection_decomp_sigma_poly} follows similarly.

\end{proof}

Defining the material \emph{contrast parameters} 
$h_i=\sigma_i/\sigma_d$ and $z_i=\sigma_d/\sigma_i=1/h_i$,
$i=1,\ldots,d$, equations \eqref{eq:avg_J} and \eqref{eq:projection_decomp_sigma_poly}  can be written as
\begin{align}
&\bsigma^*\vecE _0=\sigma_d\left\langle(h_1X_1+h_2X_2+h_{d-1}X_{d-1}+X_d)\vecE \right\rangle,
\\
&\brho^*\vecJ _0=(1/\sigma_d)\left\langle(z_1X_1+z_2X_2+h_{d-1}X_{d-1}+X_d)\vecJ \right\rangle,
\notag
\end{align}
which are functions of the $h_i$ and $z_i$. Due to the direct analogue of equation 
\eqref{eq:projection_decomp_sigma_multicomp} with equation
\eqref{eq:projection_decomp_sigma_poly} 
the analytical properties of $\bsigma^*$ and $\brho^*$ for dimensions $d\ge3$ can be 
described using methods that are direct analogues of the methods developed in
\cite{Golden:JSP-655:40} for multi-component $(n\ge3)$ media --- details will be published elsewhere.

\subsection{Uniaxial polycrystalline media}

In this section we restrict our attention to uniaxial polycrystalline media, where $\sigma_2=\sigma_3=\cdots=\sigma_d$. 
For 2D this poses no restrictions on the class of polycrystalline media. However,
for 3D this imposes a uniaxial microscopic asymmetry. 
We continue our presentation in terms of electrical conductivity and resistivity 
%
\begin{align}\label{eq:two-phase_eps}
\bsigma(\vecx ,\omega)=\sigma_1X_1(\vecx ,\omega)+\sigma_2X_2(\vecx ,\omega)\,, 
\qquad
\brho(\vecx ,\omega)=X_1(\vecx ,\omega)/\sigma_1+X_2(\vecx ,\omega)/\sigma_2\,,
\end{align}
keeping in mind the broader applicability of the method, where $X_1=R^{\,T}CR$, $X_2=R^{\,T}(I-C)R$, and 
$C=\text{diag}(\vece_1)$. Due to the direct analogue of equation 
\eqref{eq:projection_decomp_sigma_multicomp} for two-component materials $(n=2)$ 
with equations \eqref{eq:projection_decomp_sigma_poly} 
for general 2D polycrystalline media or uniaxial polycrystalline media for $d\ge3$, the analytical properties of 
$\bsigma^*$ and $\brho^*$ can be described using methods that are direct analogues of the methods 
developed in \cite{Bergman:PRL-1285:44,Milton:APL-300,Golden:CMP-473,Murphy:2015:CMS:13:4:825} 
for multi-component media. We now develop this mathematical theory.

By equations \eqref{eq:Stationary_sig} and \eqref{eq:Stationary_rho}, we have the existence of 
measurable functions $[X_i^{\,\prime}(\omega)]_{jk}$, $j,k=1,\ldots,d$, such that
$[X_i(\vecx ,\omega)]_{jk}=[X_i^{\,\prime}(\tau_{-x}\,\omega)]_{jk}$ for all
$\vecx \in\mathbb{R}^d$ and $\omega\in\Omega$. More specifically for 2D, from equation \eqref{eq:polycrystal_parameters_2D} we have the existence of a measurable function 
$\theta(\vecx ,\omega)=\theta^{\,\prime}(\tau_{-x}\,\omega)$ for all
$\vecx \in\mathbb{R}^d$ and $\omega\in\Omega$ that determines the orientation of the crystal with
$\vecx$ as an interior point. For 3D, from equation  \eqref{eq:polycrystal_parameters_3D}
we have the existence of measurable functions 
$\theta_i(\vecx ,\omega)=\theta_i^{\,\prime}(\tau_{-x}\,\omega)$, $i=1,2,3$, for all
$\vecx \in\mathbb{R}^d$ and $\omega\in\Omega$, and similarly for $d>3$.

Recall that the field $\vecE$ defined by $\bnabla\btimes\vecE=0$ is
determined only up to an arbitrary additive constant, hence the definition $\vecE=\vecE_0+\vecE_f$, where the field $\vecE_f$ is uniquely determined and the field $\vecE_0$ is arbitrary. Therefore, for the formulation of the effective parameter
problem involving $\Hs_\times$ and $\bsigma^*$, define the
coordinate system so that in~\eqref{eq:Maxwells_Equations_E} the constant
vector $\vecE _0=\langle\vecE \rangle$ is given by
$\vecE _0=E_0\,\vece _j$ for some $j=1,\ldots,d$. Similarly, for the formulation of the effective parameter
problem involving $\Hs_\bullet$ and $\brho^*$, instead, define the
coordinate system so that in~\eqref{eq:Maxwells_Equations_J} the constant
vector $\vecJ _0=\langle\vecJ \rangle$ is given by
$\vecJ _0=J_0\,\vece _j$. (Note that, in general, the vectors $\vecE_0$ and
$\vecJ_0$ do not have the same direction.)

Define the complex 
material contrast variables $h=\sigma_1/\sigma_2$, $z=\sigma_2/\sigma_1=1/h$,
and use equations
\eqref{eq:avg_J} and \eqref{eq:two-phase_eps} to write
$\bsigma^*$ and $\brho^*$ as
\begin{align}
\label{eq:eps*_rho*}
&\bsigma^*(h)\vecE _0=\sigma_2\left\langle(hX_1+X_2)\vecE \right\rangle,
\quad
\brho^*(h)\vecJ_0=(1/\sigma_1)\left\langle(X_1+hX_2)\vecJ\,\right\rangle,
\\
&\bsigma^*(z)\vecE _0=\sigma_1\left\langle(X_1+zX_2)\vecE \right\rangle,
\quad
\brho^*(z)\vecJ_0=(1/\sigma_2)\left\langle(zX_1+X_2)\vecJ\,\right\rangle.
\notag
\end{align}
which are functions of the material \emph{contrast parameters} $h$ and $z$. 
We now establish that the components $\sigma_{jk}^*(h)=\bsigma^*(h)\vece_j\bcdot\vece_k$,
$j,k=1,\ldots,d$ of $\bsigma^*(h)$ are analytic function of $h$ for $h\notin(-\infty,0]$.
\begin{theorem}\label{thm:Analytic_h}
The bilinear functional $\Phi$ in equation~\eqref{eq:Bilinear_functional_E} is coercive when the complex variable 
$h$ is off the negative real axis, $h\notin(-\infty,0]$. Moreover, the effective conductivity representation $\sigma^*_{jk}(h)$ is an analytic function of $h$ for $h\notin(-\infty,0]$.
\end{theorem}
\begin{proof}
The bounded bilinear functional $\Phi$ in equation \eqref{eq:Bilinear_functional_E}  is coercive if there exists a $\kappa>0$ such that $|\Phi(\vecxi,\vecxi)|\geq\kappa\|\vecxi\,\|^2$ for all
$\vecxi\in\Hs_\times$ such that $\|\vecxi\,\|\neq0$. From equation \eqref{eq:two-phase_eps} this is
true only if  
\begin{align}\label{eq:Coercive_Phi_h}
    \left|
    \left\langle(hX_1+X_2)\vecxi\bcdot\vecxi\;\right\rangle
    \right|\geq (\kappa/|\sigma_2|)\|\vecxi\,\|^2,
\end{align}
where $\vecxi$ is complex-valued when $h$ is complex and $\vecxi\bcdot\vecxi$
involves complex conjugation in the rightmost position.
Denote $\gamma$ by the ratio 
\begin{align} 
    \gamma=\frac{\left\langle X_1\vecxi\bcdot\vecxi\;\right\rangle}          
    {\|\vecxi\,\|^2}\,.
\end{align}
Since $X_1$ is a real-valued projection matrix, we have
$\langle X_1\vecxi\bcdot\vecxi\rangle=\langle X_1\vecxi\bcdot X_1\vecxi\rangle=\|X_1\vecxi\,\|^2\geq0$.
Consequently, $\gamma$ is real and $\gamma\geq0$, 
Moreover, since $X_1$ is bounded by 1 in operator norm we also have $\gamma\le1$.
Consequently, we have that $0\leq\gamma\leq1$. Since
$X_2=I-X_1$, the coercivity condition in \eqref{eq:Coercive_Phi_h} can now be
written as
\begin{align}\label{eq:Coercive_Phi_h_beta}
    |h\gamma+1-\gamma|\geq\kappa/|\sigma_2|>0.
\end{align}

Let $h=h_r+\imath h_i$, where $h_r$ and $h_i$ are the real and imaginary
parts of the complex number $h$, respectively. Since 
$|h\gamma+1-\gamma|^2=|1+\gamma(h_r-1)|^2+\gamma^2|h_i|^2$, the formula in
\eqref{eq:Coercive_Phi_h_beta} always holds when $h_i\neq0$. If $h$ is
real, then $h\gamma+1-\gamma=0$ if and only if $h=1-1/\gamma$  
for $0<\gamma\leq1$, i.e. for $h$ on the negative real axis. In conclusion,
equation~\eqref{eq:Coercive_Phi_h_beta} holds if and only if $h$ is
\emph{off} of the negative real axis including zero, i.e.,  $h\notin(-\infty,0]$
\cite{Golden:CMP-473}. For any such complex value of $h$, $\Phi$ is
coercive and there exists a unique solution $\vecE _f$ to equation
\eqref{eq:Bilinear_functional_E}, hence to equation
\eqref{eq:Maxwells_Equations_E} (weakly). 
%
%
By differentiation with respect to
$h$, one deduces easily 
that $\vecE _f$ is analytic
in $h$ off the negative real axis with values in
$\Hs_\times$. Therefore by equation~\eqref{eq:eps*_rho*},
$\sigma^*_{jk}(h)$ is analytic off the negative real axis in the complex
$h$-plane. 
\end{proof}

\begin{corollary}
The effective resistivity representation $\rho^*_{jk}(h)$ is an analytic function of $h$ for $h\notin(-\infty,0]$.
The effective conductivity and resistivity representations $\sigma^*_{jk}(z)$ and $\rho^*_{jk}(z)$ are analytic functions of $z$ for $z\notin(-\infty,0]$.
\label{thm:Analytic_z}
\end{corollary}
\begin{proof}
We note that the transformation $h=1/z$ maps the interval $(-\infty,0)$ in the $h$-plane to the interval $(-\infty,0)$
in the $z$-plane. Therefore, in studies of the analytic properties of bulk transport coefficients, one can focus on $h$\
with $z(h)\not\in(-\infty,0)$ when $h\not\in(-\infty,0)$ \cite{Murphy:JMP:063506}. One representation of the bulk
transport coefficient $\sigma^*_{jk}(h)$ or $\sigma^*_{jk}(z)$ (and similarly for $\rho^*_{jk}$) in equation \eqref{eq:m_h} is often more convenient than the other in studies of the limits
$h\to0$ or $h\to\infty$, depending on which limit is being considered \cite{Murphy:JMP:063506}. With these 
observations, the remainder of the proof is analogous to Theorem \eqref{thm:Analytic_h}.
\end{proof}

\section{Spectral theory for effective parameters}
\label{sec:Representation_formulas}
In this section we show that the abstract Hilbert space formulation for the effective 
parameter problem given in Section \ref{sec:Homogenizaiton_and_Hilbert} leads
to Stieltjes integral representations for $\bsigma^*$ and $\brho^*$, involving spectral
measures of self-adjoint operators. In Section \ref{sec:Hilbert_operators} we introduce Hilbert
spaces and linear operators encountered in the spectral theory for effective parameters, and prove that 
each operator is self-adjoint on an appropriate Hilbert space. In Section \ref{sec:Resolvent_Formulas}
we provide resolvent formulas for electric and current fields involving the self-adjoint operators. In
Section \ref{sec:Stieltjes_integral} we utilize the resolvent formulas and the spectral theorem 
\cite{Stone:64,Schmudgen:2012:2012942602} to provide Stieltjes integral representations for $\bsigma^*$ 
and $\brho^*$, involving spectral measures of the self-adjoint operators.

\subsection{Hilbert space and self-adjoint operators}
\label{sec:Hilbert_operators}
We start this section by defining Hilbert spaces used in the remainder of this manuscript. Recall 
the Hilbert spaces $\Hs_0=L^2(\Omega,P)$ and $\Hs=\bigotimes_{i=1}^d \Hs_0$, 
first defined beneath  equation \eqref{eq:Li_curl_div}, each with 
\emph{sesquilinear} inner products. The $\Hs_0$-inner-product is defined by 
$\langle\xi,\zeta\rangle_0=\langle\xi\bar{\zeta}\rangle$ with 
$\langle\zeta ,\xi\rangle_0=\overline{\langle\xi,\zeta \rangle}_0$, where 
$\langle\cdot\rangle$ denotes ensemble average over $\Omega$ with respect to the
measure $P$. This inner-product induces the norm $\langle\xi,\xi\rangle_0=\|\xi\|_0^2\,$.
The $\Hs$-inner-product is defined by 
$\langle\vecxi,\bzeta \rangle=
\langle\vecxi\bcdot\bzeta \rangle=
\sum_i\langle\xi_i,\zeta_i\rangle_0$, 
with
$\langle\bzeta ,\vecxi\rangle=
\overline{\langle\vecxi,\bzeta \rangle}$.
This inner-product induces the norm 
$\langle\vecxi,\vecxi\rangle=\|\vecxi\|^2=\langle|\vecxi|^2\rangle$.

Also recall the operators $L_i$, $i=1,\ldots,d$, defined above equation \eqref{eq:Li}, that are densely 
defined with domains $\Ds_i\subset\Hs_0$, and recall the definition of the closed linear space $\Ds=\cap_{i=1}^d\Ds_i$, where $\Ds\subset\Hs_0$.
The operators encountered in this section are symmetric with respect to the $\Hs$-inner-product,
but involve the operators $L_i$ and the operators $L_i L_j$,
$i,j=1,\ldots,d$. The following spaces $\Hs^1_0$, $\Hs^2_0$, $\Hs^1=\bigotimes_{i=1}^d \Hs^1_0$, and
$\Hs^2=\bigotimes_{i=1}^d \Hs^2_0$ are normed vector spaces, where
\begin{align}\label{eq:H1_H2}
&\Hs^1_0=\{\xi\in\Ds:\|L_i\xi\|_0<\infty \text{ for all } i=1,\ldots,d, \ \langle\xi\rangle=0 \},
\\\notag
&\Hs^2_0=\{\xi\in\Ds:\|L_iL_j\xi\|_0<\infty \text{ for all } i,j=1,\ldots,d,  \ \langle\xi\rangle=0 \}.
\notag
\end{align}
Every normed vector space can be isometrically embedded onto a dense vector subspace of some Banach space, 
where this Banach space is called a completion of the normed space \cite{Folland:99:RealAnalysis}. 
Consequently, we can take these normed vector spaces to be complete Hilbert spaces, where 
$\Hs^1_0$ and $\Hs^2_0$ are equipt with the $\Hs_0$-inner-product $\langle\cdot,\cdot\rangle_0$, while
$\Hs^1$ and $\Hs^2$ are equipt with the $\Hs$-inner-product $\langle\cdot,\cdot\rangle$.

In this manuscript, we will show that a self-adjoint operator $X_1\Gamma X_1$ --- a composition of the 
projection operator $X_1$ with another projection operator $\Gamma$ --- appears in a resolvent 
representation for the field $X_1\vecE$, which leads to Stieltjes integral representations for the 
components of the effective conductivity matrix $\bsigma^*$ involving spectral measures for the self-adjoint operator $X_1\Gamma X_1$. Therefore, orthogonal properties of projection operators play a central role in 
our analysis. 
Denote by $\Dc(T)$, $\Rc(T)$ and $\Kc(T)$ the \emph{domain}, \emph{range}, and \emph{kernel} 
of an operator $T$, respectively. The following lemma establishes key properties of projection 
operators, e.g., $X_1$ and $\Gamma$ \cite{Folland:99:RealAnalysis}.
\begin{lemma}
\label{lem:orthogonal}	
Let $\Hc$ be a Hilbert space with inner-product $\langle\cdot,\cdot\rangle$. 
If $\xi,\zeta\in\Hc$, we say that $\xi$ is \emph{orthogonal} to $\zeta$ and write $\xi\perp\zeta$ if $\langle\xi,\zeta\rangle=0$. If $\Ms\subset \Hc$ we define
\begin{align}
    \Ms^\perp=\{\xi\in\Hc:\langle\xi,\zeta\rangle=0 \text{ for all } \zeta\in\Ms\},
\end{align}
which is a closed subspace of $\Hc$. If $\Ms$ is a closed subspace of $\Hc$ then 
$\Hc=\Ms\oplus\Ms^\perp$, i.e., each $\xi\in\Hc$ can be uniquely written as $\xi=\zeta+\gamma$
with $\zeta\in\Ms$ and $\gamma\in\Ms^\perp$ with $\zeta\perp\gamma$. If $\Ms$ is a subset of $\Hc$, then $(\Ms^\perp)^\perp$
is the smallest closed subspace of $\Hc$ containing $\Ms$. Finally, let $\Ms$ be a closed subspace of $\Hc$
and for $\xi\in\Hc$, let $P\xi$ be the element of $\Ms$ such that $\xi-P\xi\in\Ms^\perp$, then $P:\Hc\mapsto\Hc$
is a linear operator satisfying $P^*=P$ and $P^2=P$. Conversely, if $P$ is a linear operator satisfying 
$P^*=P$ and $P^2=P$ then $\Rc(P)$ is closed and $P$ is the orthogonal projection onto 
$\Rc(P)$ \cite{Folland:99:RealAnalysis}.
\end{lemma}

The following lemma introduces the projection operator $\Gamma$ and establishes that the operators
$\Gamma$ and $X_1\Gamma X_1$ are self-adjoint on $\Hs^1$.
\begin{theorem}\label{thm:self-adjoint_X1GammaX1}
Assume the projection matrix defined in equation \eqref{eq:two-phase_eps} satisfies $X_1:\Hs^1\mapsto\Hs^1$.
Consider the generalized gradient operator $\bnabla=(L_1,L_2,\ldots,L_d)^T$ and its adjoint $\bnabla^*$.
The operator $\Gamma=\bnabla (\bnabla^*\bnabla)^{-1}\bnabla^*$ is the orthogonal projection 
onto $\Rc(\bnabla)$ and satisfies $\Gamma:\Hs^1\mapsto\Rc(\bnabla)\cap\Hs^1$. 
The operators $\Gamma$ and $X_1\Gamma X_1$ are self-adjoint on the Hilbert space $\Hs^1$.
The spectrum $\Sigma$ of the self-adjoint operator $X_1\Gamma X_1$ satisfies $\Sigma\subseteq[0,1]$.
\end{theorem}
\begin{proof}
The abstract gradient operator $\bnabla\xi=(L_1\xi,L_2\xi,\ldots,L_d\xi)^T$ is well defined for $\xi\in\Hs^1_0$.
We therefore take $\Dc(\bnabla)=\Hs^1_0$ and $\Rc(\bnabla)\subseteq\Hs$,  i.e., $\bnabla:\Hs^1_0\mapsto\Hs$.
Since the operators $L_i$ are skew-symmetric \cite{Papanicolaou:RF-835} on $\Hs^1_0$, by the definition 
of the adjoint operator $\bnabla^*$  \cite{Schmudgen:2012:2012942602,Stone:64}, we have
\begin{align}
    \langle\bnabla\xi,\bzeta\rangle=
    \sum_{i=1}^d\langle L_i\xi,\zeta_i\rangle_0=
    -\sum_{i=1}^d\langle\xi,L_i\zeta_i\rangle_0=
    \langle\xi,\bnabla^*\bzeta\rangle_0\,.
\end{align}
Therefore, the adjoint operator $\bnabla^*$ is given by a generalization of the negative divergence operator 
with $\bnabla^*\bzeta=-\sum_{i=1}^d L_i\zeta_i$, which is well defined for $\bzeta\in\Hs^1$. 
We therefore take $\Dc(\bnabla^*)=\Hs^1$ and $\Rc(\bnabla^*)\subseteq\Hs_0$,  
i.e., $\bnabla:\Hs^1\mapsto\Hs_0$.

It follows that the ``modulus square'' operator $\bnabla^*\bnabla$ is given by an elliptical operator, with 
$\bnabla^*\bnabla\xi=-\sum_i L_i^2\xi$, which is well defined for $\xi\in\Hs^2_0$.
We therefore take $\Dc(\bnabla^*\bnabla)=\Hs^2_0$ and $\Rc(\bnabla^*\bnabla)\subseteq\Hs_0$,  
i.e., $\bnabla^*\bnabla:\Hs^2_0\mapsto\Hs_0$, where $\Rc(\bnabla^*\bnabla)\subseteq\Rc(\bnabla^*)\subseteq\Hs_0$.
Since $\langle\bnabla^*\bnabla\xi,\xi\rangle_0=\langle\bnabla\xi,\bnabla\xi\rangle=\|\bnabla\xi\|^2\ge0$
for all $\xi\in\Hs^2_0$, the operator $\bnabla^*\bnabla$ is positive on $\Hs^2_0$. The ``modulus square''
operator of an arbitrary densely defined operator always has a positive self-adjoint extension 
\cite{Zoltzan:2023:03081087,Schmudgen:2012:2012942602}. Therefore, since the $L_i$ are closed and densely 
defined~\cite{Papanicolaou:RF-835}, the operator $\bnabla^*\bnabla$ can be considered to be a positive 
self-adjoint operator on $\Hs^2_0$.

We now establish that the operator $\bnabla^*\bnabla$ has a well defined inverse satisfying
$(\bnabla^*\bnabla)^{-1}:\Rc(\bnabla^*\bnabla)\mapsto\Hs^2_0$. From equation
\eqref{eq:weak_derivative}, within an inner-product, we have $L_i=\partial_i$, where $\partial_i$ denotes classical
partial differentiation, so $-\bnabla^*\bnabla$ is given by the
Laplacian $\Delta=\nabla^2$. We now show that when the polycrystalline material is contained within a spatial  domain $\Vc$ that is bounded in one or more spatial directions, the operator $\Delta$ with homogeneous 
boundary conditions is one-to-one on $\Hs^2_0$ and therefore has a well defined inverse. Indeed, the
Poincar\'{e} inequality then states that for all $\xi\in\Hs^2_0$ with $\xi\neq0$, there exists a constant 
$\beta>0$ such that 
\begin{align}
    \langle -\Delta\xi,\xi\rangle_0=
    \langle \bnabla\xi\bcdot\bnabla\xi\rangle\ge
    \beta\langle\xi,\xi\rangle_0=
    \beta\|\xi\|_0^2>0.
\end{align}
Now if $f\neq g$ on $\Hs^2_0$ but $\Delta f=\Delta g$ then
\begin{align}
    \langle \Delta(f-g),(f-g)\rangle_0=
    \langle (\Delta f-\Delta g),(f-g)\rangle_0=
    \langle 0,(f-g)\rangle_0=
    0.
\end{align}
However, by the Poincar\'{e} inequality, $\langle -\Delta(f-g),(f-g)\rangle_0\ge\beta\|f-g\|_0^2>0$, as $f\neq g$.
This contradiction establishes that $\Delta$, hence $\bnabla^*\bnabla$  is one-to-one on $\Hs^2_0$ and 
therefore has a well defined inverse satisfying $(\bnabla^*\bnabla)^{-1}:\Hs_0\mapsto\Hs^2_0$.

For the case of an unbounded domain $\Vc=\mathbb{R}^d$, we note that if 
$\xi\in\Kc(\bnabla^*\bnabla)$, where $\Kc(\bnabla^*\bnabla)\subset\Hs^2_0$,
then
\begin{align}
    0=\langle\bnabla^*\bnabla\xi,\xi\rangle=
    \langle\bnabla\xi,\bnabla\xi\rangle=
    ||\bnabla\xi||^2,
\end{align}
hence $\bnabla\xi=0$ almost everywhere (a.e.) \cite{Folland:99:RealAnalysis}. Therefore, $\xi$ is constant
a.e.. 
However, by the definition of $\Hs^2_0$, $\xi\in\Hs^2_0$ implies $\langle\xi\rangle=0$,  hence $\xi=0$ a.e.. 
Consequently, 
$\Kc(\bnabla^*\bnabla)=\{0\}$ which implies that $\bnabla^*\bnabla$ has a well defined inverse 
\cite{Zoltzan:2023:03081087,Schmudgen:2012:2012942602},
$(\bnabla^*\bnabla)^{-1}:\Rc(\bnabla^*\bnabla)\mapsto\Hs^2_0$.
The theory of Greens functions provides an explicit representation for the operator $(-\Delta)^{-1}$ in terms
of convolution with the Greens function for the Laplacian \cite{Stakgold:BVP:2000}.

We now establish that the operator $(\bnabla^*\bnabla)^{-1}$ is self-adjoint on $\Hs_0\,$. 
Since the operator $\bnabla^*\bnabla$ is unbounded on $\Hs^2_0$  the operator 
$(\bnabla^*\bnabla)^{-1}$ is bounded \cite{Murphy:ADSTPF-2017,Stone:64,Schmudgen:2012:2012942602,Stakgold:BVP:2000}. Since $\bnabla^*\bnabla$ is symmetric, $(\bnabla^*\bnabla)^{-1}$ is also symmetric. Indeed, denoting $T=\bnabla^*\bnabla$, $T\xi=f$, and
$T\zeta=g$, we have
\begin{align}
    \langle T^{-1}f,g\rangle_0=
    \langle \xi,T\zeta\rangle_0=
    \langle T\xi,\zeta\rangle_0=
    \langle f,T^{-1}g\rangle_0\,.
\end{align}
The bounded symmetric operator $(\bnabla^*\bnabla)^{-1}$ has bounded self-adjoint extensions 
\cite{Zoltzan:2023:03081087,Schmudgen:2012:2012942602} and can therefore be considered to be a 
self-adjoint operator on $\Hs_0$.

The gradient operator $\bnabla$ with $\Dc(\bnabla)=\Hs^2_0$ satisfies $\bnabla:\Hs^2_0\mapsto\Hs^1$.
Our analyses in this proof 
has established that 
the operator $\Gamma=\bnabla(\bnabla^*\bnabla)^{-1}\bnabla^*$ with $\Dc(\Gamma)=\Hs^1$ is linear 
and symmetric, $\Gamma^*=\Gamma$, and has $\Rc(\Gamma)=\Hs^1$, i.e., $\Gamma:\Hs^1\mapsto\Hs^1$.  
More specifically, recalling that $\Rc(\bnabla^*\bnabla)\subseteq\Rc(\bnabla^*)$,

\begin{align}
    \bnabla^*:\Hs^1\mapsto\Rc(\bnabla^*)\cap\Hs_0,
    \quad
    (\bnabla^*\bnabla)^{-1}:\Rc(\bnabla^*\bnabla)\cap\Hs_0\mapsto\Hs^2_0,
    \quad
    \bnabla:\Hs^2_0\mapsto\Rc(\bnabla)\cap\Hs^1.
\end{align}
It is clear that $\Gamma^2=\Gamma$ and  $\Gamma\bnabla\xi=\bnabla\xi$, weakly. From Lemma 
\ref{lem:orthogonal}, it follows that the operator $\Gamma$ is the orthogonal projection 
onto $\Rc(\bnabla)$ and $\Rc(\bnabla)$ is a closed subset of $\Hs^1$. Therefore $\Gamma$ has 
bounded operator norm $\|\Gamma\|\le1$ on $\Hs^1$ \cite{Folland:99:RealAnalysis}.

We now establish that the operators $\Gamma$ and $X_1\Gamma X_1$ are positive self-adjoint operators on the Hilbert space $\Hs^1$. 
Since $X_1$ is a matrix-valued projection operator, the operator $X_1\Gamma X_1$ has bounded operator 
norm $\|X_1\Gamma X_1\|\le1$ on $\Hs^1$  \cite{Folland:99:RealAnalysis}. The components of $X_1$ 
can be discontinuous functions of the spatial variable $\vecx\in\mathbb{R}^d$; we assume that 
$X_1:\Hs^1\mapsto\Hs^1$, hence $X_1\Gamma X_1:\Hs^1\mapsto\Hs^1$. Since $\Gamma^2=\Gamma$, 
it is a positive operator on $\Hs^1$ \cite{Schmudgen:2012:2012942602}. Moreover, since $X_1$ is a
real-symmetric projection matrix satisfying $X_1^*=X_1^T=X_1$,
and $X_1^2=X_1$, the operator $X_1\Gamma X_1$ is also positive on
$\Hs^1$: 
\begin{align}
    \langle X_1\Gamma X_1\vecxi,\vecxi\rangle=
    \langle \Gamma X_1\vecxi,\Gamma X_1\vecxi\rangle=
    \|\Gamma X_1\vecxi\|^2\ge0.
\end{align}
The bounded positive symmetric operators $\Gamma$ and $X_1\Gamma X_1$ have bounded positive 
self-adjoint extensions \cite{Zoltzan:2023:03081087,Schmudgen:2012:2012942602} and can therefore 
be considered to be positive self-adjoint operators on $\Hs^1$. Since the self-adjoint operator 
$X_1\Gamma X_1$ is positive, the spectrum satisfies \cite{Schmudgen:2012:2012942602,Stone:64}  $\Sigma\subseteq[0,\infty)$. The spectrum $\Sigma$ of $X_1\Gamma X_1$ also satisfies \cite{Schmudgen:2012:2012942602,Stone:64}  
$\Sigma\subseteq[-\|X_1\Gamma X_1\|,\|X_1\Gamma X_1\|]$.  
Therefore, since $\|X_1\Gamma X_1\|\le1$, we have $\Sigma\subseteq[0,1]$.
\end{proof}

In the following two Corollaries, we collect relevant parts of the proof of Theorem \ref{thm:self-adjoint_X1GammaX1} that will be needed later in the
transcript.
\begin{corollary}\label{cor:inverse_Laplacian}
The operator $\bnabla^*\bnabla:\Hs^2_0\mapsto\Hs_0$ is one-to-one when the spatial domain $\Vc$ is
bounded in one or more directions and $\bnabla^*\bnabla$ has homogeneous boundary conditions. When 
$\Vc=\mathbb{R}^d$, we have $\Kc(\bnabla^*\bnabla)=\{0\}$. Therefore, in either case, $\bnabla^*\bnabla$ 
has a well defined inverse $(\bnabla^*\bnabla)^{-1}:\Rc(\bnabla^*\bnabla)\mapsto\Hs^2_0$.
\end{corollary}
\begin{corollary}\label{eq:Laplacian_trivial_kernel}
The operator $\bnabla^*\bnabla$ defined with domain $\Dc(\bnabla^*\bnabla)=\Hs^2_0$ has the trivial kernel 
$\Kc(\bnabla^*\bnabla)=\{0\}$.
\end{corollary}

We now turn our attention from the projection $\Gamma$ onto $\Rc(\bnabla)$ towards a projection $\Upsilon$ 
onto $\Rc(\bnabla\btimes)$ that arises in Stieltjes
integration representations for the components of the effective resistivity matrix $\brho^*$, involving
spectral measures of a self-adjoint operator that incorporates $\Upsilon$.
We start by providing a result that relates the ranges and kernels of an operator $T$ and its adjoint $T^*$.
\begin{lemma}
\label{lem:generalization_FTLA}
Let $\Hc$ be a Hilbert space with inner-product $\langle\cdot,\cdot\rangle$ and 
let $T$ be a \emph{closed} linear operator on $\Hc$ with adjoint $T^*$. Then, 
\begin{align}
    \label{eq:perp1}	
    \Kc(T)=\Rc(T^*)^\perp,
    \\
    \label{eq:perp2}
    (\Rc(T)^\perp)^\perp=\Kc(T^*)^\perp.
\end{align}
\end{lemma}
\begin{proof}
For the proof of equation \eqref{eq:perp1}, let $\xi\in\Kc(T)$ and $\zeta\in\Rc(T^*)$. 
Then $\zeta=T^*\gamma$ for some $\gamma$ in $\Hc$. Therefore,
\begin{align}
    \langle\xi,\zeta\rangle=
    \langle\xi,T^*\gamma\rangle=
    \langle T\xi,\gamma\rangle=
    \langle 0,\gamma\rangle=0,
\end{align}
as $\xi\in\Kc(T)$, hence $\Kc(T)\subseteq\Rc(T^*)^\perp$. On the other hand, let $\xi\in\Rc(T^*)^\perp$, 
then for all $\zeta\in\Hc$
\begin{align}
    \langle T\xi,\zeta\rangle=
    \langle \xi,T^*\zeta\rangle=0,
\end{align}
as $T^*\zeta\in\Rc(T^*)$. Therefore, we must have $T\xi=0$ \cite{Folland:99:RealAnalysis}, hence
$\Rc(T^*)^\perp\subseteq\Kc(T)$. Consequently, we have $\Kc(T)=\Rc(T^*)^\perp$, which establishes 
equation \eqref{eq:perp1}.

For the proof of equation \eqref{eq:perp2}, since $T$ is closed we have $T^{**}=T$ 
\cite{Schmudgen:2012:2012942602}. Therefore, from equation \eqref{eq:perp1} we have 
$\Kc(T^*)=\Rc(T)^\perp$, hence $\Kc(T^*)^\perp=[\Rc(T)^\perp]^\perp$, where $[\Rc(T)^\perp]^\perp$ 
is the smallest closed set containing $\Rc(T)$ \cite{Folland:99:RealAnalysis}, which establishes 
equation  \eqref{eq:perp2}.
\end{proof}

Define the generalized curl operator
\begin{align}
\label{eq:generalized_curl}
\bnabla\btimes&=\left[
\begin{array}{rrr}
    O&-L_3&L_2\\
    L_3&O&-L_1\\  
    -L_2&L_1&O
\end{array}
\right]
\;\text{ for 3D,} 
\\
\bnabla\btimes&=[-L_2,L_1]
\;\text{ for 2D} 
\notag
\end{align}
where $O$ denotes the null operator satisfying $O\xi=0$ for all $\xi\in\Hs_0$. We first focus on the 3D setting.
Since $\bnabla\btimes$ is given in terms of the skew-symmetric Levi-Civita tensor and the operators $L_i$
are skew-symmetric on $\Hs^1_0$, the curl operator is symmetric, $(\bnabla\btimes)^*=\bnabla\btimes$.
The operation $\bnabla\btimes\bxi$ is well defined for $\bxi\in\Hs^1$. We therefore take 
$\Dc(\bnabla\btimes)=\Hs^1$ and $\Rc(\bnabla\btimes)\subseteq\Hs$,  i.e., 
$\bnabla\btimes:\Hs^1\mapsto\Hs$.
Since the operators $L_i$ commute on $\Hs^2_0$ \cite{Papanicolaou:RF-835,Golden:CMP-473}, a direct
calculation shows that $\bnabla^*\bnabla\btimes\bxi=0$, weakly, for all $\bxi\in\Hs^2$ and 
$\bnabla\btimes\bnabla\xi=0$, weakly, for all $\xi\in\Hs^2_0$. Therefore,
\begin{align}\label{eq:curl_grad_div_Ran_Ker}
\Rc(\bnabla\btimes)\subseteq\Kc(\bnabla^*),
\quad
\Rc(\bnabla)\subseteq\Kc(\bnabla\btimes).
\end{align}
Since each operator $L_i$, $i=1,\ldots,d$, is closed with densely defined domains $\Ds_i\in\Hs_0$, 
the operator $\bnabla$ with domain $\Hs^1_0$ and 
the operator $\bnabla\btimes$ with domain $\Hs^1$ are
also closed, so Lemma \ref{lem:generalization_FTLA} can be applied to 
$\bnabla$ and $\bnabla\btimes$.

We now summarize the results of Lemma \ref{lem:generalization_FTLA} in the context of the operators 
$\bnabla$ and $\bnabla\btimes$ --- taking in account the symmetry $\bnabla\btimes^*=\bnabla\btimes\,$.
\begin{align}\label{eq:kernel_range}
\Kc(\bnabla)=\Rc(\bnabla^*)^\perp,
\quad
\Kc(\bnabla^*)=\Rc(\bnabla)^\perp,
\quad
\Kc(\bnabla\btimes)=\Rc(\bnabla\btimes)^\perp.
%
%
\end{align}
We are now ready to prove the following corollary of Theorem \ref{thm:self-adjoint_X1GammaX1}.
\begin{corollary}\label{cor:self-adjoint_X1UpsilonX1}
Assume the projection matrix defined in equation \eqref{eq:two-phase_eps} satisfies $X_1:\Hs^1\mapsto\Hs^1$.
Consider the generalized curl operator $\bnabla\btimes$ in equation \eqref{eq:generalized_curl} and its adjoint 
$\bnabla\btimes^*=\bnabla\btimes$.
The operator 
$\Upsilon=\bnabla\btimes(\bnabla\btimes\bnabla\btimes)^{-1}\bnabla\btimes$ 
is the orthogonal projection onto $\Rc(\bnabla\btimes)$ and satisfies
$\Upsilon:\Hs^1\mapsto\Rc(\bnabla\btimes)\cap\Hs^1$. The 
operators $\Upsilon$ and $X_1\Upsilon X_1$ are self-adjoint on the Hilbert space $\Hs^1$.
The spectrum of the self-adjoint operator $X_1\Upsilon X_1$ is contained in the interval $[0,1]$.
\end{corollary}
\begin{proof}
Since $\Dc(\bnabla\btimes)=\Hs^1$ and $\Rc(\bnabla\btimes)\subseteq\Hs$,  i.e., 
$\bnabla\btimes:\Hs^1\mapsto\Hs$, it follows that the ``modulus square''
operator $(\bnabla\btimes)^*\bnabla\btimes=\bnabla\btimes\bnabla\btimes$ is 
well defined for $\Dc(\bnabla\btimes\bnabla\btimes)=\Hs^2$ and $\Rc(\bnabla\btimes\bnabla\btimes)\subseteq\Hs$,  i.e., 
$\bnabla\btimes\bnabla\btimes:\Hs^2\mapsto\Hs$.
Analogous to the proof of Theorem \ref{thm:self-adjoint_X1GammaX1}, it can be shown that the operator $\bnabla\btimes\bnabla\btimes$ 
is positive on $\Hs^2$. Moreover, the ``modulus square'' operator $(\bnabla\btimes)^*\bnabla\btimes$ can 
be considered to be a positive self-adjoint operator on $\Hs^2$ 
\cite{Zoltzan:2023:03081087,Schmudgen:2012:2012942602}. 

We now establish 
that the operator $\bnabla\btimes\bnabla\btimes$ has a well defined inverse
$(\bnabla\btimes\bnabla\btimes)^{-1}$ satisfying
$(\bnabla\btimes\bnabla\btimes)^{-1}:\Rc(\bnabla\btimes\bnabla\btimes)\mapsto\Hs^2$. Since $\Gamma$ 
is a projection operator, its range $\Rc(\bnabla)$ is closed. Therefore, by Lemma \ref{lem:orthogonal} we have $\Hs^2=\Rc(\bnabla)\oplus\Rc(\bnabla)^\perp$. From equation \eqref{eq:kernel_range} it follows that 
$\Hs^2=\Rc(\bnabla)\oplus\Kc(\bnabla^*)$, therefore each $\bxi\in\Hs^2$ can be uniquely written as
$\bxi=\bzeta+\bgamma$ with $\bzeta\in\Rc(\bnabla)$, $\bgamma\in\Kc(\bnabla^*)$, and $\bzeta\perp\bgamma$. From equation \eqref{eq:curl_grad_div_Ran_Ker} we have $\bnabla\btimes\bnabla\btimes\bzeta=0$ as $\Rc(\bnabla)\subseteq\Kc(\bnabla\btimes)$, hence 
$\bnabla\btimes\bnabla\btimes\bxi=\bnabla\btimes\bnabla\btimes\bgamma$. 
Since $\bgamma\in\Hs^2$ and the operators $L_i$ commute on $\Hs^2_0$ \cite{Papanicolaou:RF-835}, 
a direct calculation shows that
\begin{align}\label{eq:curl_curl}
    \bnabla\btimes\bnabla\btimes\bgamma=\text{diag}(\bnabla^*\bnabla)\bgamma-\bnabla\bnabla^*\bgamma.
\end{align}
However, $\bgamma\in\Kc(\bnabla^*)$ implies that $\bnabla\bnabla^*\bgamma=0$. Therefore the 
action of the operator $\bnabla\btimes\bnabla\btimes$ with domain 
$\Dc(\bnabla\btimes\bnabla\btimes)=\Kc(\bnabla^*)$ is given by
$\bnabla\btimes\bnabla\btimes\bgamma=\text{diag}(\bnabla^*\bnabla)\bgamma$, for all 
$\bgamma\in\Kc(\bnabla^*)$. Since 
$\Rc(\bnabla\btimes\bnabla\btimes)\subseteq\Rc(\bnabla\btimes)\subseteq\Kc(\bnabla^*)$, 
it follows from Corollary \ref{cor:inverse_Laplacian} that $\bnabla\btimes\bnabla\btimes$ has a well defined inverse
$(\bnabla\btimes\bnabla\btimes)^{-1}$ satisfying
$(\bnabla\btimes\bnabla\btimes)^{-1}:\Rc(\bnabla\btimes\bnabla\btimes)\mapsto\Hs^2$.

Analogous to the proof of Theorem \ref{thm:self-adjoint_X1GammaX1}, the bounded symmetric operator $(\bnabla\btimes\bnabla\btimes)^{-1}$ can be considered to be a bounded positive self-adjoint operator on 
$\Hs^2$. Moreover, the operator $\Upsilon=\bnabla\btimes(\bnabla\btimes\bnabla\btimes)^{-1}\bnabla\btimes$ 
with $\Dc(\Upsilon)=\Hs^1$
is the orthogonal projection onto $\Rc(\bnabla\btimes)$ and satisfies $\Upsilon^*=\Upsilon$, 
$\Upsilon^2=\Upsilon$, $\|\Upsilon\|\le1$, and $\Upsilon:\Hs^1\mapsto\Hs^1$. More specifically, 
recalling that $\Rc(\bnabla\btimes\bnabla\btimes)\subseteq\Rc(\bnabla\btimes)\subseteq\Kc(\bnabla^*)$,
\begin{align}
    \bnabla\btimes:\Hs^1\mapsto\Rc(\bnabla\btimes)\cap\Hs,
    \quad
    (\bnabla\btimes\bnabla\btimes)^{-1}:\Rc(\bnabla\btimes\bnabla\btimes)\cap\Hs\mapsto\Hs^2,
    \quad
    \bnabla\btimes:\Hs^2\mapsto\Rc(\bnabla\btimes)\cap\Hs^1.
\end{align}
Moreover, the operator $X_1\Upsilon X_1$ satisfies 
$X_1\Upsilon X_1:\Hs^1\mapsto\Hs^1$. Both $\Upsilon$ and $X_1\Upsilon X_1$ can be considered to 
be positive self-adjoint operators with operator norms less that or equal to unity. Analogous to the proof of 
Theorem \ref{thm:self-adjoint_X1GammaX1}, the spectrum of the self-adjoint operator $X_1\Upsilon X_1$ is contained
in the interval $[0,1]$.
\end{proof}

The following corollary immediately follows from Theorem \ref{thm:self-adjoint_X1GammaX1} and Corollary \ref{cor:self-adjoint_X1UpsilonX1}.
\begin{corollary}
If all the projection operators $X_i$, $i=1,2$, defined below equation \eqref{eq:two-phase_eps} satisfy
$X_i:\Hs^1\mapsto\Hs^1$, then all the operators $X_i\Gamma X_i$ and $X_i\Upsilon X_i$, $i=1,2$, 
with domain defined to be $\Hs^2$ are self-adjoint.
\end{corollary}

We now prove a Theorem regarding the setting of a two-component composite material briefly
introduced in the introduction of Section \ref{sec:Analytic_properties}, which follows from the proofs of Theorem \ref{thm:self-adjoint_X1GammaX1} and Corollary \ref{cor:self-adjoint_X1UpsilonX1} and the remarkable 
symmetries between equations \eqref{eq:projection_decomp_sigma_multicomp} and 
\eqref{eq:projection_decomp_sigma_poly} for the two-component setting and uniaxial polycrystalline setting. 
\begin{theorem}
\label{thm:self-adjoint_two-component}
Consider a locally isotropic two-component composite material with local conductivity matrix given 
by $\bsigma(\vecx,\omega)=\sigma(\vecx,\omega)I$ and resistivity matrix 
$\brho(\vecx,\omega)=\rho(\vecx,\omega)I$, where $I$ 
denotes the identity operator on $\mathbb{R}^d$. The scalar valued functions $\sigma(\vecx,\omega)$ and 
$\rho(\vecx,\omega)$ are given by 
\begin{align}
    \sigma(\vecx,\omega)=\sigma_1\chi_1(\vecx,\omega)+\sigma_2\chi_2(\vecx,\omega),
    \quad
    \rho(\vecx,\omega)=\rho_1\chi_1(\vecx,\omega)+\rho_2\chi_2(\vecx,\omega).
\end{align}
Here, $\rho_i=1/\sigma_i$, $i=1,2$, and the characteristic function, $\chi_i$, takes the value 
$\chi_i(\vecx,\omega)=1$ when the material has component $i$ located at $\vecx$ for 
$\omega\in\Omega$ and $\chi_i(\vecx,\omega)=0$ otherwise and satisfies 
\begin{align}\label{eq:chi_projection}
    \chi_1^*=\chi_1\,,
    \quad
    \chi_i \chi_j=\chi_i\delta_{ij}\,,
    \quad
    \chi_1+\chi_2=1\,.
\end{align}
If all the projection operators $\chi_i$, $i=1,2$ satisfy
$\chi_i:\Hs^1\mapsto\Hs^1$, then all the operators $\chi_i\Gamma \chi_i$ and $\chi_i\Upsilon \chi_i$, $i=1,2$, 
with domain defined to be $\Hs^2$ are self-adjoint. 

Moreover, as the matrix-valued projection operators $X_i$, $i=1,2$, also satisfy the properties displayed in 
equation \eqref{eq:chi_projection}
\begin{align}\label{eq:X_projection}
    X_1^*=X_1\,,
    \quad
    X_i X_j=X_i\delta_{ij}\,,
    \quad
    X_1+X_2=I\,,
\end{align}
also assuming $X_i:\Hs^1\mapsto\Hs^1$, every theorem, lemma, and corollary stated in Sections \ref{sec:Analytic_properties} and  \ref{sec:Representation_formulas} involving the $X_i$ for the
uniaxial polycrystalline setting, also holds for the two-component setting with the simple exchange of $X_i$ 
with $\chi_i$, $i=1,2$. 
\end{theorem}
\begin{proof}
The only assumptions made in the proofs of Theorem \ref{thm:self-adjoint_X1GammaX1} and Corollary 
\ref{cor:self-adjoint_X1UpsilonX1} regarding the projection operators $X_i$, $i=1,2$, are those given 
in equation \eqref{eq:X_projection} --- along with $X_i:\Hs^1\mapsto\Hs^1$. Since the projection operators 
$\chi_i$, $i=1,2$ satisfy the analogous properties shown in equation \eqref{eq:chi_projection}, 
if all the projection operators $\chi_i$, $i=1,2$ satisfy $\chi_i:\Hs^1\mapsto\Hs^1$,
it follows that all the operators $\chi_i\Gamma \chi_i$ and $\chi_i\Upsilon \chi_i$, $i=1,2$, 
with domain defined to be $\Hs^2$ are self-adjoint. From the symmetries given in equations \eqref{eq:projection_decomp_sigma_multicomp} and 
\eqref{eq:projection_decomp_sigma_poly}, it follows that every theorem, lemma, and corollary stated in Sections 
\ref{sec:Analytic_properties} and  \ref{sec:Representation_formulas} involving the $X_i$ for the uniaxial polycrystalline setting also 
holds for the two-component setting with the simple exchange of $X_i$ with $\chi_i$, $i=1,2$.

\end{proof}

We now show that the operators $\Gamma$ and $\Upsilon$ are mutually orthogonal projections and together
provide a resolution of the identity.
\begin{theorem}\label{thm:resolution_identity}
On $\Hs^2$ the orthogonal projections $\Gamma$ and $\Upsilon$ are mutually orthogonal operators 
satisfying the following resolution of the identity formula
\begin{align}\label{eq:mutually_orthogonal}
    \Gamma+\Upsilon=I.
\end{align}
More specifically, let $\Dc(\bnabla)=\Hs^2$ and $\Dc(\bnabla\btimes)=\Hs^2$ then
\begin{align}\label{eq:orthogonal_decomp_grad_curl}
    \Hs^2=\Rc(\bnabla)\oplus\Rc(\bnabla\btimes)\,,
\end{align}
and equation \eqref{eq:curl_grad_div_Ran_Ker} can be refined to
\begin{align}\label{eq:curl_grad_div_Ran_Ker_equality}
    \Rc(\bnabla\btimes)=\Kc(\bnabla^*),
    \quad
    \Rc(\bnabla)=\Kc(\bnabla\btimes).
\end{align}
\end{theorem}

\begin{proof}
From equation \eqref{eq:curl_grad_div_Ran_Ker} we have
\begin{align}
    \Upsilon\Gamma\bxi=\Gamma\Upsilon\bxi=0,
\end{align}
for all $\bxi\in\Hs^2$, i.e., $\Upsilon$ and $\Gamma$ are mutually orthogonal projection operators. 	
By Theorem \ref{thm:self-adjoint_X1GammaX1}, $\Gamma$ is a projection onto $\Rc(\bnabla)$. Therefore,
by Lemma \ref{lem:orthogonal} and equation \eqref{eq:kernel_range}, we have 
$\Hs^1=\Rc(\bnabla)\oplus\Kc(\bnabla^*)$. By equation \eqref{eq:curl_grad_div_Ran_Ker} we have
$\Rc(\bnabla\btimes)\subseteq\Kc(\bnabla^*)$. By Corollary \ref{cor:self-adjoint_X1UpsilonX1}, $\Upsilon$
is a projection onto $\Rc(\bnabla\btimes)$. Therefore by Lemma \ref{lem:orthogonal},
$\Kc(\bnabla^*)=\Rc(\bnabla\btimes)\oplus\Rc(\bnabla\btimes)^\perp$. 
Combining everything, we have
\begin{align}
    \Hs^2=\Rc(\bnabla)\oplus\Rc(\bnabla\btimes)\oplus\Rc_0\,,
\end{align}
where we have denoted $\Rc_0=\Rc(\bnabla\btimes)^\perp$ to emphasize that $\Rc_0$ is mutually orthogonal  
to both $\Rc(\bnabla)$ and $\Rc(\bnabla\btimes)$.

Let $\bgamma\in\Rc_0$ so that 
$\bgamma\in\Kc(\bnabla^*)$, $\bgamma\in\Rc(\bnabla)^\perp$ and $\bgamma\in\Rc(\bnabla\btimes)^\perp$.
Therefore, from equation \eqref{eq:curl_curl}, if $\bgamma\in\Hs^2$,
\begin{align}\label{eq:curl_curl_resolution_identity}
    \langle\text{diag}(\bnabla^*\bnabla)\bgamma,\bgamma\rangle=
    \langle\bnabla\btimes\bnabla\btimes\bgamma,\bgamma\rangle+
    \langle\bnabla\bnabla^*\bgamma,\bgamma\rangle=0.
\end{align}
Therefore, for each component $\gamma_i$, $i=1,\ldots,d$, of the vector $\bgamma$, we have
\begin{align}
    0=
    \langle\bnabla^*\bnabla\gamma_i,\gamma_i\rangle=
    \langle\bnabla\gamma_i,\bnabla\gamma_i\rangle=
    \|\bnabla\gamma_i\|^2.
\end{align}
Consequently, from Corollary \ref{eq:Laplacian_trivial_kernel}, $\gamma_i=0$ for all $i=1,\ldots,d$, hence 
$\Rc_0=\{0\}$.

It follows that equation \eqref{eq:orthogonal_decomp_grad_curl} holds	
which implies equation \eqref{eq:mutually_orthogonal}. Moreover, this analysis demonstrates that when we 
define $\Dc(\bnabla)=\Hs^2$ and $\Dc(\bnabla\btimes)=\Hs^2$, then equation \eqref{eq:curl_grad_div_Ran_Ker}
can be refined to equation \eqref{eq:curl_grad_div_Ran_Ker_equality}.
\end{proof}

We conclude this section with the following theorem which will be used
in Section \ref{sec:Resolvent_Formulas} to provide resolvent representations for $X_1\vecE$ and $X_1\vecJ$
involving the self-adjoint operators $X_1\Gamma X_1$ and $X_1\Upsilon X_1$, respectively, etc.
\begin{theorem}
\label{thm:H2_decomp_Hx_H.}	
Let $\Dc(\bnabla)=\Hs^2$ and $\Dc(\bnabla\btimes)=\Hs^2$ and re-define the Hilbert spaces in 
equation \eqref{eq:curlfreeHilbert} as
\begin{align}\label{eq:curlfreeHilbert_H2}
    \Hs_\times=\Hs^2\cap\Kc(\bnabla\btimes),
    \quad
    \Hs_\bullet=\Hs^2\cap\Kc(\bnabla^*).
\end{align}  
Then Theorem \ref{thm:existenceE} and Corollary \ref{cor:existenceJ} still hold and 
$\Hs_\times=\Rc(\bnabla)$ and  $\Hs_\bullet=\Rc(\bnabla\btimes)$, hence
\begin{align}\label{eq:DecompH2_times_bullet}
    \Hs^2=\Hs_\times\oplus\Hs_\bullet\,.
\end{align}
Moreover, $\Gamma$ is the orthogonal projection onto 
$\Hs_\times$ and the electric field $\vecE$ in Theorem 
\ref{thm:existenceE} satisfies 
\begin{align}\label{eq:E_Gamma}
    \vecE=\vecE_0+\vecE_f\,,
    \quad 
    \Gamma\vecE_0=0,
    \quad
    \Gamma\vecE_f=\vecE_f\,.
\end{align}
On the other hand, $\Upsilon $ is the orthogonal 
projection onto $\Hs_\bullet$ and the current field $\vecJ$ in Corollary 
\ref{cor:existenceJ} satisfies 
\begin{align}\label{eq:J_Upsilon}
    \vecJ=\vecJ_0+\vecJ_f\,,
    \quad
    \Upsilon\vecJ_0=0,
    \quad
    \Upsilon\vecJ_f=\vecJ_f\,.
\end{align}
\end{theorem}
\begin{proof}
The proofs of Theorem \ref{thm:existenceE} and Corollary \ref{cor:existenceJ} only take in account the 
Hilbert spaces $\Hs_\times$ and $\Hs_\bullet$ defined in equation \eqref{eq:curlfreeHilbert} \emph{implicitly} through the statement of the weak 
variational formulation of Maxwell's equations in \eqref{eq:Weak_Alternate_Maxwell_E} and 
\eqref{eq:Weak_Curl_Free_Variational_Form_E} involving $\bsigma$ and $\vecE$, and the analogous weak 
variational problem involving $\brho$ and $\vecJ$. Therefore, the proofs of Theorem \ref{thm:existenceE} and 
Corollary \ref{cor:existenceJ} still hold by instead using the definition of the 
Hilbert spaces $\Hs_\times$ and $\Hs_\bullet$ in equation \eqref{eq:curlfreeHilbert_H2}, which simply include
additional regularity by using the base Hilbert space $\Hs^2$ instead of $\Hs$. 

Equation \eqref{eq:DecompH2_times_bullet} follows from equations \eqref{eq:orthogonal_decomp_grad_curl}, 
\eqref{eq:curl_grad_div_Ran_Ker_equality}, and \eqref{eq:curlfreeHilbert_H2}. In particular, when
$\Dc(\bnabla)=\Hs^2$ and $\Dc(\bnabla\btimes)=\Hs^2$ we have $\Hs_\times=\Rc(\bnabla)$ and  $\Hs_\bullet=\Rc(\bnabla\btimes)$. It follows that, $\Gamma$ with 
domain $\Dc(\Gamma)=\Hs^2$ is the orthogonal projection onto $\Hs_\times$ and  $\Upsilon$ with 
domain $\Dc(\Upsilon)=\Hs^2$ is the orthogonal projection onto $\Hs_\bullet\,$. Therefore, from Theorem \ref{thm:existenceE} we have $\vecE=\vecE_0+\vecE_f$, where $\vecE_0$ is a constant and 
$\vecE_f\in\Hs_\times$, so $\Gamma\vecE_0=0$ and $\Gamma\vecE_f=\vecE_f$, hence 
$\Gamma\vecE=\vecE_f$. This establishes equation \eqref{eq:E_Gamma}. Similarly, from Corollary \ref{cor:existenceJ} we have $\vecJ=\vecJ_0+\vecJ_f$, where $\vecJ_0$ is a constant and 
$\vecJ_f\in\Hs_\bullet$, so $\Upsilon\vecJ_0=0$ and $\Upsilon\vecJ_f=\vecJ_f$, hence 
$\Upsilon\vecJ=\vecJ_f$. This establishes equation \eqref{eq:J_Upsilon}.
\end{proof}

\subsection{Resolvent formulas}
\label{sec:Resolvent_Formulas}

Define the complex material contrast variables $s=1/(1-h)$ and $t=1/(1-z)=1-s$, and use equations
\eqref{eq:projection_decomp_sigma_poly} and \eqref{eq:two-phase_eps} to write
$\bsigma$ and $\brho$ in~\eqref{eq:two-phase_eps} as
%
\begin{align}\label{eq:sig_rho_s_t}
\bsigma=\sigma_2(I-X_1/s)=\sigma_1(I-X_2/t),
\qquad
\brho=\sigma_2^{-1}(I-X_1/t)=\sigma_1^{-1}(I-X_2/s).
\end{align}
Equations~\eqref{eq:avg_J} and \eqref{eq:eps*_rho*},
imply the components 
$\sigma^*_{jk}=\bsigma^*\vece _j\bcdot\vece _k$ and
$\rho^*_{jk}=\brho^*\vece _j\bcdot\vece _k$
of the matrices $\bsigma^*$ and $\brho^*$ have the following functional 
forms 
\begin{align}\label{eq:Eff_Cond_Tens_Def}
\sigma^*_{jk}
&=\sigma_2(\delta_{jk} -\langle X_1\vecE\bcdot\vece_k \rangle/(sE_0))
=\sigma_1(\delta_{jk} -\langle X_2\vecE\bcdot\vece_k \rangle/(tE_0)),
\\
\rho^*_{jk}
&=\sigma_2^{-1}(\delta_{jk} -\langle X_1\vecJ\bcdot\vece_k\rangle/(tJ_0))
=\sigma_1^{-1}(\delta_{jk} -\langle X_2\vecJ\bcdot\vece_k\rangle/(sJ_0)),
\notag
\end{align}
which are functions of the material \emph{contrast parameters} $h$ and $z$.
We define the matrix-valued functions $\mathbf{m}(h)=\bsigma^*/\sigma_2$,
$\mathbf{w}(z)=\bsigma^*/\sigma_1$, $\tilde{\mathbf{m}}(h)=\sigma_1\brho^*$, and
$\tilde{\mathbf{w}}(z)=\sigma_2\brho^*$ with components  
\begin{align}\label{eq:m_h}
m_{jk}(h)=\sigma_{jk}^*/\sigma_2, 
\quad
w_{jk}(z)=\sigma_{jk}^*/\sigma_1,
\quad
\tilde{m}_{jk}(h)=\sigma_1\rho_{jk}^*, 
\quad
\tilde{w}_{jk}(z)=\sigma_2\rho_{jk}^*,
\end{align}
The above analysis~\cite{Golden:CMP-473} demonstrates that the
dimensionless functions $m_{jk}(h)$ and $\tilde{m}_{jk}(h)$ in
\eqref{eq:m_h} are analytic off the negative real axis in the
$h$-plane, while $w_{jk}(z)$ and $\tilde{w}_{jk}(z)$ are analytic off
the negative real axis in the $z$-plane. By equations
\eqref{eq:eps*_rho*}, each
take the corresponding upper half plane to the upper half plane and
are therefore examples of Herglotz functions
\cite{Deift:2000:RMT,Golden:CMP-473}.

We now derive a resolvent representation for the field $X_1\vecE$ appearing in the linear functional 
$\langle X_1\vecE\bcdot\vece_k \rangle$ in equation \eqref{eq:Eff_Cond_Tens_Def}.
\begin{lemma}\label{lem:resolvent_X1E}
For a uniaxial composite material, the field $X_1\vecE$ has the following resolvent representation involving 
the self-adjoint operator
$X_1\Gamma X_1$.
\begin{align}\label{eq:Resolvent_representations_X1E}
    X_1\vecE =s(sI-X_1\Gamma X_1)^{-1}X_1\vecE _0 ,
    \quad
    s\in\mathbb{C}\backslash[0,1].
\end{align}
\end{lemma}

\begin{proof}
From Theorem \ref{thm:existenceE}, the electric field satisfies $\vecE =\vecE _0+\vecE _f$, where $\vecE _0$ is
constant and $\vecE _f\in\Hs_\times$. From Theorem \ref{thm:H2_decomp_Hx_H.} we have $\Gamma\vecE_0=0$
and $\Gamma\vecE_f=\vecE_f$, hence $\Gamma\vecE=\vecE_f$. From Corollary \ref{cor:existenceJ}, the 
current field satisfies $\vecJ =\vecJ _0+\vecJ _f$, where $\vecJ _0$ is constant and $\vecJ _f\in\Hs_\bullet$. 
From Theorem \ref{thm:H2_decomp_Hx_H.} we have $\Upsilon\vecJ_0=0$
and $\Upsilon\vecJ_f=\vecJ_f$, hence $\Upsilon\vecJ=\vecJ_f$. From Theorem \ref{thm:resolution_identity} we 
have $\Gamma\vecJ_f=0$ hence $\Gamma\vecJ=0$, which together with $\vecJ=\bsigma\vecE$ in \eqref{eq:Maxwells_Equations_E0} and equation \eqref{eq:sig_rho_s_t} yield
\begin{align}\label{eq:Proj_rep_Ef}
    \vecE _f=\frac{1}{s}\Gamma X_1\vecE. 
\end{align}
Adding $\vecE_0$ to both sides of \eqref{eq:Proj_rep_Ef} then applying $X_1$ to the left of both sides, using 
$X_1^2=X_1$, and rearranging establishes equation \eqref{eq:Resolvent_representations_X1E}, 
which holds \cite{Schmudgen:2012:2012942602,Stone:64} for $s\in\mathbb{C}$ not in the 
spectrum of the operator $X_1\Gamma X_1$, which by Theorem \ref{thm:self-adjoint_X1GammaX1} is contained 
in the interval $[0,1]$.
\end{proof}

The proofs of the following two corollaries are analogous to the proof of Lemma \ref{lem:resolvent_X1E}.
\begin{corollary}
\label{cor:resolvent_X2E_X1J_X2J}
For a uniaxial composite material, the fields $X_i\vecE$ and $X_i\vecJ$, $i=1,2$, have the following 
resolvent representations involving the self-adjoint operators $X_i\Gamma X_i$ and $X_i\Upsilon X_i$, 
$i=1,2$.
\begin{align}\label{eq:Resolvent_representations_X}
    &X_1\vecE =s(sI-X_1\Gamma X_1)^{-1}X_1\vecE _0 ,
    \quad
    s\in\mathbb{C}\backslash[0,1].
    \\\notag
    &X_2\vecE =t(tI-X_2\Gamma X_2)^{-1}X_2\vecE _0 ,
    \quad
    t\in\mathbb{C}\backslash[0,1].
    \\\notag
    &X_1\vecJ=t(tI-X_1\Upsilon X_1)^{-1}X_1\vecJ_0 ,
    \quad
    t\in\mathbb{C}\backslash[0,1].
    \\\notag
    &X_2\vecJ=s(sI-X_2\Upsilon X_2)^{-1}X_2\vecJ_0 ,
    \quad
    s\in\mathbb{C}\backslash[0,1].
    \\\notag
\end{align}
\end{corollary}

From Theorem \ref{thm:self-adjoint_two-component} and Corollary \ref{cor:resolvent_X2E_X1J_X2J}  
we also have the following result for the setting of a two-component material.
\begin{corollary}
\label{cor:resolvent_chi2E_chi1J_chi2J}
For a two-component composite material, the fields $\chi_i\vecE$ and $\chi_i\vecJ$, $i=1,2$, have the 
following resolvent representations involving the self-adjoint operators $\chi_i\Gamma \chi_i$ and 
$\chi_i\Upsilon \chi_i$, $i=1,2$.
\begin{align}\label{eq:Resolvent_representations_chi}
    &\chi_1\vecE =s(sI-\chi_1\Gamma \chi_1)^{-1}\chi_1\vecE _0 ,
    \quad
    s\in\mathbb{C}\backslash[0,1].
    \\\notag
    &\chi_2\vecE =t(tI-\chi_2\Gamma \chi_2)^{-1}\chi_2\vecE _0 ,
    \quad
    t\in\mathbb{C}\backslash[0,1].
    \\\notag
    &\chi_1\vecJ=t(tI-\chi_1\Upsilon \chi_1)^{-1}\chi_1\vecJ_0 ,
    \quad
    t\in\mathbb{C}\backslash[0,1].
    \\\notag
    &\chi_2\vecJ=s(sI-\chi_2\Upsilon \chi_2)^{-1}\chi_2\vecJ_0 ,
    \quad
    s\in\mathbb{C}\backslash[0,1].
    \\\notag
\end{align}
\end{corollary}

\subsection{Stieltjes integral representations for effective parameters}
\label{sec:Stieltjes_integral}
Consider the functional formulas in equation \eqref{eq:Eff_Cond_Tens_Def} for the 
components $\sigma^*_{jk}$ and $\rho^*_{jk}$, $j,k=1,\ldots,d$, of the effective conductivity $\bsigma^*$ 
and resistivity $\brho^*$ matrices, involving the fields $X_i\vecE$ and $X_i\vecJ$, $i=1,2$. Also, consider the
resolvent representations for the fields $X_i\vecE$ and $X_i\vecJ$, $i=1,2$, given in equation \eqref{eq:Resolvent_representations_X}, involving the self-adjoint operators 
$X_i\Gamma X_i$ and $X_i\Upsilon X_i$, $i=1,2$. In this section, we show that the linear functionals in  
\eqref{eq:Eff_Cond_Tens_Def} involving the resolvent representations of fields in \eqref{eq:Resolvent_representations_X} have Stieltjes integral representations involving spectral measures
of the self-adjoint operators.

We start with a review of relevant aspects of the spectral
theorem for bounded self-adjoint operators in Hilbert space.
The spectrum $\Sigma$ of a self-adjoint operator $T$ on a Hilbert space
$\Hc$ is real-valued~\cite{Reed-1980,Stone:64}. We will assume that
$T$ is a \emph{bounded} operator. In this case, its spectral radius
equal to its operator norm $\|T\|$~\cite{Reed-1980}, i.e.,  
\begin{align}\label{eq:Spectral_Radius_Phi}
\Sigma\subseteq[-\|T\|,\|T\|\,].
\end{align}
The spectral theorem states that there is a one-to-one
correspondence between the self-adjoint operator $T$ and a family of
self-adjoint projection operators $\{Q(\lambda)\}_{\lambda\in\Sigma}$ --- the resolution of
the identity --- that satisfies~\cite{Reed-1980,Stone:64} 
%
\begin{align}\label{eq:Res_Identity_limits}
\lim_{\lambda\to\,\inf{\Sigma}}Q(\lambda)=0, \quad
\lim_{\lambda\to\,\sup{\Sigma}}Q(\lambda)=I.
\end{align}
Here, $0$ and $I$ denote the null and identity operators on
$\Hc$. Furthermore, the \emph{complex-valued} function of the spectral 
variable $\lambda$ defined by $\mu_{\psi\varphi}(\lambda)=\langle Q(\lambda)\psi,\varphi\,\rangle$ is 
of bounded variation for all
$\psi,\varphi\in\Hc$~\cite{Stone:64}. By the sesquilinearity of the inner-product
and the fact that the projection operator $Q(\lambda)$ is self-adjoint, the
function $\mu_{\psi\varphi}(\lambda)$ satisfies
$\mu_{\varphi\psi}(\lambda)=\overline{\mu}_{\psi\varphi}(\lambda)$. Moreover, the function $\mu_{\psi\psi}(\lambda)$
is real-valued and positive $\mu_{\psi\psi}(\lambda)=\langle Q(\lambda)\psi,\psi\rangle=\langle
Q(\lambda)\psi,Q(\lambda)\psi\rangle=\|Q(\lambda)\psi\,\|^2\geq0$.  Since the function $\mu_{\psi\varphi}(\lambda)$ is of bounded variation, it has an associated
Stieltjes
measure~\cite{Stieltjes:1995,Stone:64,Folland:99:RealAnalysis}    
$\d\mu_{\psi\varphi}(\lambda)=\d\langle Q(\lambda)\psi,\varphi\rangle$ which we will denote by $\mu_{\psi\varphi}$ for brevity.

The spectral theorem also provides an operational calculus in Hilbert
space which yields integral representations involving the Stieltjes
measure $\mu_{\psi\varphi}$. A summary of the relevant details are as   
follows. Let $Y(\lambda)$ and $Z(\lambda)$ be arbitrary complex-valued functions
and denote by $\Ds(Y)$ the set of all $\psi\in\Hc$ such that $Y\in
L^2(\mu_{\psi\psi})$, i.e., $Y$ is square integrable on the set $\Sigma$ with
respect to the \emph{positive} measure $\mu_{\psi\psi}$, and similarly define 
$\Ds(Z)$. Then $\Ds(Y)$ and $\Ds(Z)$ are linear manifolds and there
exist linear operators denoted by $Y(T)$ and $Z(T)$ with domains
$\Ds(Y)$ and $\Ds(Z)$, respectively, which are defined in terms of the
following Radon--Stieltjes integrals~\cite{Stone:64}  
\begin{align}\label{eq:Spectral_Theorem}
\langle Y(T)\psi,\varphi\rangle&=\int_\Sigma Y(\lambda)\,\d\mu_{\psi\varphi}(\lambda), \qquad
\hspace{1em}
\forall \, \psi\in\mathscr{D}(Y), \ \varphi\in\Hc,  
\\
\langle Y(T)\psi,Z(T)\varphi\rangle&=\int_\Sigma Y(\lambda)\overline{Z}(\lambda)\,\d\mu_{\psi\varphi}(\lambda),
\quad
\forall \, \psi\in\mathscr{D}(Y), \ \varphi\in\mathscr{D}(Z),
\notag
\end{align}
where the integration in~\eqref{eq:Spectral_Theorem} is over the
spectrum $\Sigma$ of $T$~\cite{Reed-1980,Stone:64}.

The mass $\mu^0_{\psi\varphi}=\int_\Sigma\d\mu_{\psi\varphi}(\lambda)$ of the \emph{Stieltjes
measure} $\mu_{\psi\varphi}$ satisfies~\cite{Stone:64} 
$\mu^0_{\psi\varphi}=\lim_{\lambda\to\sup\Sigma}\mu_{\psi\varphi}(\lambda)-\lim_{\lambda\to\inf\Sigma}\mu_{\psi\varphi}(\lambda)$. Consequently, equation~\eqref{eq:Res_Identity_limits}, the continuity of inner-products 
\cite{Folland:99:RealAnalysis}, and the Cauchy-Schwartz inequality yield \cite{Schmudgen:2012:2012942602,Stone:64}
\begin{align}\label{eq:Mass_General}
\mu^0_{\psi\varphi}=\int_\Sigma\d\langle Q(\lambda)\psi,\varphi\,\rangle=\langle\psi,\varphi\rangle,
\qquad
|\mu^0_{\psi\varphi}|\leq\|\psi\|\,\|\varphi\|<\infty.
\end{align}
Equation~\eqref{eq:Mass_General} demonstrates that the measures
$\mu_{\psi\varphi}$ are \emph{finite measures}, i.e., they have bounded
mass~\cite{Stone:64}. The spectral theorem demonstrates that the moments 
$\mu^n_{\psi\varphi}$, $n=1,2,\ldots$, of the measure are given by
\begin{align}\label{eq:Moments_General}
\mu^n_{\psi\varphi}=\langle T^n\psi,\varphi\rangle=
\int_\Sigma\lambda^n\d\langle Q(\lambda)\psi,\varphi\,\rangle.
\end{align}
The function $Y(\lambda)$ associated with the
resolvent operator $R_s=(sI-T)^{-1}$, for $s\not\in\Sigma$, is given by
$Y(\lambda)=1/(s-\lambda)$. For $s$ a positive distance away from $\Sigma$, i.e.,
$|s-\lambda|>\delta>0$ for some $\delta>0$ and all $\lambda\in\Sigma$, the function $Y(\lambda)$ is
uniformly bounded in modulus. In this case, equation~\eqref{eq:Mass_General} implies that $\mathscr{D}(Y)$
coincides with the entire Hilbert space $\Hc$. To preclude the degenerate case of a spectral measure with 
moments satisfying $\mu^n_{\psi\varphi}=0$ for all $n=1,2,\ldots$, we will only consider $\psi,\varphi\not\in\Kc(T)$.

We are now ready to state the main theorem in this section.
\begin{theorem}\label{thm:X1E}
Consider the self-adjoint operator $\bM=X_1\Gamma X_1$ with domain 
$\Dc(\bM)=\{\bxi\in\Hs^2:X_1\bxi\neq0\}$, with $X_1:\Hs^2\mapsto\Hs^2$, 
where the condition $X_1\bxi\neq0$ is to preclude the degenerate case 
described above. Then the components $\sigma^*_{jk}=\sigma_2 m_{jk}(h)$,  
$j,k=1,\ldots,d$, of the matrix $\bsigma^*$ have 
the following Stieltjes integral representations
\begin{align}\label{eq:Stieltjes_F}
    &m_{jk}(h)=\delta_{jk}-F_{jk}(s), 
    \quad
    F_{jk}(s)
    =\langle(sI-X_1\Gamma X_1)^{-1}X_1\vece _j\bcdot\vece _k\rangle
    =\int_0^1\frac{\d\mu_{jk}(\lambda)}{s-\lambda}\,.
    \quad
    s\in\mathbb{C}\backslash[0,1].
\end{align}
Here, $\mu_{jk}$ is a  \emph{spectral measure} associated
with the self-adjoint random operator $X_1\Gamma X_1$; the measure $\mu_{jk}$ is a
\emph{signed measure} for $j\ne k$ and the measure $\mu_{kk}$ is a 
\emph{positive measure}.
There exists a family of self-adjoint projection operators 
$\{Q(\lambda)\}_{\lambda\in\Sigma}$ --- the resolution of the identity --- that is  uniquely determined 
by $X_1\Gamma X_1$ and satisfies  $Q(0)=0$ and $Q(1)=I$. The real-valued function  
$\mu_{jk}(\lambda)=\langle Q(\lambda)X_1\vece _j,\vece _k\,\rangle$ 
for $\lambda\in\Sigma$ is of bounded variation and the Stieltjes measure in 
\eqref{eq:Stieltjes_F} is given by    
$\d\mu_{jk}(\lambda)=\d\langle Q(\lambda)X_1\vece _j,\vece _k\rangle$. The 
operators $Q(\lambda)$ and $X_1$ commute. 
\end{theorem}
\begin{proof}
Since $X_1^2=X_1$, we we have $\|X_1\vece_j\|\le1$, for all $j=1,\ldots,d$. 
Moreover, the matrix components of $X_1$ are infinitely differentiable trigonometric functions except for
along the crystal boundaries. We assume that these boundaries are regular enough to ensure
$X_1:\Hs^2\mapsto\Hs^2$, hence $X_1\vece_j\in\Hs^2$. Since the vector $X_1\vece_j$ contains terms that are
strictly positive almost everywhere, we have $X_1\vece_j\neq0$, for all $j=1,\ldots,d$. Consequently,
$X_1\vece_j\in\Dc(\bM)$. Taking $Y(\lambda)=(s-\lambda)^{-1}$ in \eqref{eq:Spectral_Theorem} with 
$s\in\mathbb{C}\backslash[0,1]$ ensures the domain of the operator $\Ds(Y)$ coincides with the
entire Hilbert space $\Hs^2$.	

Substituting the resolvent representation for the field $X_1\vecE$ displayed in equation 
\eqref{eq:Resolvent_representations_X} into the linear functional for $\sigma^*_{jk}(s)$ in equation
\eqref{eq:Eff_Cond_Tens_Def}, using the first formula in equation \eqref{eq:Spectral_Theorem} with 
$Y(\lambda)=(s-\lambda)^{-1}$, using the first formula in equation \eqref{eq:m_h}, and noting by Theorem \ref{thm:self-adjoint_X1GammaX1} the spectrum $\Sigma$ of the operator $\bM=X_1\Gamma X_1$ satisfies $\Sigma\subseteq[0,1]$, establishes equation \eqref{eq:Stieltjes_F}, where we have denoted $\mu_{\vece_j\vece_k}=\mu_{jk}$. 
Since $X_1^2=X_1$ and $X_1^*=X_1^T=X_1$, we also have 
\begin{align}
    \label{eq:two_reps_X1}
    \int_0^1\frac{\d\langle Q(\lambda)X_1\vece_k\bcdot\vece_k\rangle(\lambda)}{s-\lambda}=
    \langle X_1\vecE\cdot\vece_k\rangle/(sE_0)=
    \langle X_1\vecE\cdot X_1\vece_k\rangle/(sE_0)=
    \int_0^1\frac{\d\langle Q(\lambda)X_1\vece_k\bcdot X_1\vece_k\rangle(\lambda)}{s-\lambda}\,.
\end{align}
Due to the linearity of Radon-Stieltjes integrals \cite{Schmudgen:2012:2012942602,Stone:64} in the function $\mu_{\phi\psi}(\lambda)$, e.g., in equation
\eqref{eq:Spectral_Theorem} $\int Y(\lambda)\d(\mu+\alpha)(\lambda)=\int Y(\lambda)\d\mu(\lambda) +\int Y(\lambda)\d\alpha(\lambda)$, 
from equation \eqref{eq:two_reps_X1} we have 
\begin{align}
    \int_0^1\frac{\d\left(
        \langle Q(\lambda)X_1\vece_k\bcdot\vece_k\rangle-
        \langle Q(\lambda)X_1\vece_k\bcdot X_1\vece_k\rangle
        \right)(\lambda)}{s-\lambda}=0,
    \text{ for all }
    s\in\mathbb{C}\backslash[0,1]\,.
\end{align}
By the Stieltjes-Perron inversion theorem, a measure is uniquely determined by its Stieltjes transform on 
$\mathbb{C}\backslash\mathbb{R}$ \cite{Schmudgen:2012:2012942602,Stone:64}. It follows that 
$\langle Q(\lambda)X_1\vece_k\bcdot\vece_k\rangle=\langle Q(\lambda)X_1\vece_k\bcdot X_1\vece_k\rangle$
almost everywhere, hence
\begin{align}
    \langle Q(\lambda)X_1\vece_k\bcdot\vece_k\rangle=
    \langle Q(\lambda)X_1^2\vece_k\bcdot\vece_k\rangle=
    \langle X_1 Q(\lambda)X_1\vece_k\bcdot\vece_k\rangle=
    \langle Q(\lambda)X_1\vece_k\bcdot X_1\vece_k\rangle,
\end{align}
which implies that the operators $Q(\lambda)$ and $X_1$ commute. 
\end{proof}

From equation~\eqref{eq:Mass_General}, the mass $\mu_{jk}^0$ of the measure $\mu_{jk}$ is given by 
\begin{align}\label{eq:Mass_mu}
\mu_{jk}^0=\langle X_1\vece _j\bcdot\vece _k\rangle, \quad
\mu_{kk}^0=\langle |X_1\vece _k|^2\rangle,
\end{align}
where the second equality follows from the fact that $X_1$ is a
real-symmetric projection matrix.  
The statistical average $\langle |X_1\vece _k|^2\rangle$ in~\eqref{eq:Mass_mu}
can be thought of as the ``mean orientation,'' or as the percentage of
crystallites oriented in the $k^{\text{th}}$ direction.
Since the \emph{projection} matrix $X_1$ is bounded by one in operator norm, the
Cauchy--Schwartz inequality and~\eqref{eq:Mass_mu} imply that
$0\leq\mu_{kk}^0\leq1$. By the spectral theorem, the moments $\mu_{jk}^n$, $n=1,2,\ldots$, 
of the measure $\mu_{jk}$ satisfy
\begin{align}\label{eq:Moments_mu}
\mu_{jk}^n=\int_0^1\lambda^n\mu_{jk}(\d\lambda)
=\langle [X_1\Gamma X_1]^nX_1\vece _j\bcdot\vece _k\rangle,
\quad n=0,1,2,\ldots.,
\end{align}
where the second equality follows from the first formula of the
spectral theorem in \eqref{eq:Spectral_Theorem} with
$F(\lambda)=\lambda^n$, $\psi=X_1\vece _j$, and $\phi=\vece _k$.

The proof of the following corollary is analogous to the proof of Theorem \ref{thm:X1E}.
\begin{corollary}\label{eq:Stieltjes_parameters}
The Herglotz functions in equation \eqref{eq:m_h} have the following Stieltjes integral representations, involving
spectral measures of self-adjoint operators which, in turn, involve the projection operators $X_i$, $i=1,2$. 

\begin{align}\label{eq:Stieltjes_G_X}
    &m_{jk}(h)=\delta_{jk}-F_{jk}(s), 
    \qquad
    F_{jk}(s)
    =\langle(sI-X_1\Gamma X_1)^{-1}X_1\vece _j\bcdot\vece _k\rangle
    =\int_0^1\frac{\d\mu_{jk}(\lambda)}{s-\lambda}\,,
    \\
    &w_{jk}(z)=\delta_{jk}-G_{jk}(t), 
    \qquad
    G_{jk}(t)=\langle(tI-X_2\Gamma X_2)^{-1}X_2\vece _j\bcdot\vece _k\rangle
    =\int_0^1\frac{\d\alpha_{jk}(\lambda)}{t-\lambda}\,,
    \notag \\
    &\tilde{m}_{jk}(h)=\delta_{jk}-E_{jk}(s), 
    \qquad
    E_{jk}(s)=\langle(sI-X_2\Upsilon X_2)^{-1}X_2\vece _j\bcdot\vece _k\rangle=\int_0^1\frac{\d\eta_{jk}(\lambda)}{s-\lambda}\,,
    \notag \\
    &\tilde{w}_{jk}(z)=\delta_{jk}-H_{jk}(t), 
    \qquad
    H_{jk}(t)=\langle(tI-X_1\Upsilon X_1)^{-1}X_1\vece _j\bcdot\vece _k\rangle=\int_0^1\frac{\d\kappa_{jk}(\lambda)}{t-\lambda}\,.
    \notag
\end{align}
\end{corollary}

From Theorem \ref{thm:self-adjoint_two-component} we also have the following result for the setting of 
a two-component material.
\begin{corollary}
For the setting of a two-component composite material, the Herglotz functions in equation \eqref{eq:m_h} have 
the following Stieltjes integral representations, involving spectral measures of self-adjoint operators which, in turn, involve the projection operators $\chi_i$, $i=1,2$,
associated with the setting of a two-component composite material. We emphasize that, while we use the same notation for the spectral measures in equations \eqref{eq:Stieltjes_G_X} and  \eqref{eq:Stieltjes_G_chi}, the
measures are not equal and are denoted by the same symbol for the sake of convention.
\begin{align}\label{eq:Stieltjes_G_chi}
    &m_{jk}(h)=\delta_{jk}-F_{jk}(s), 
    \qquad
    F_{jk}(s)
    =\langle(sI-\chi_1\Gamma \chi_1)^{-1}\chi_1\vece _j\bcdot\vece _k\rangle
    =\int_0^1\frac{\d\mu_{jk}(\lambda)}{s-\lambda}\,,
    \\
    &w_{jk}(z)=\delta_{jk}-G_{jk}(t), 
    \qquad
    G_{jk}(t)=\langle(tI-\chi_2\Gamma \chi_2)^{-1}\chi_2\vece _j\bcdot\vece _k\rangle
    =\int_0^1\frac{\d\alpha_{jk}(\lambda)}{t-\lambda}\,,
    \notag \\
    &\tilde{m}_{jk}(h)=\delta_{jk}-E_{jk}(s), 
    \qquad
    E_{jk}(s)=\langle(sI-\chi_2\Upsilon \chi_2)^{-1}\chi_2\vece _j\bcdot\vece _k\rangle=\int_0^1\frac{\d\eta_{jk}(\lambda)}{s-\lambda}\,,
    \notag \\
    &\tilde{w}_{jk}(z)=\delta_{jk}-H_{jk}(t), 
    \qquad
    H_{jk}(t)=\langle(tI-\chi_1\Upsilon \chi_1)^{-1}\chi_1\vece _j\bcdot\vece _k\rangle=\int_0^1\frac{\d\kappa_{jk}(\lambda)}{t-\lambda}\,.
    \notag
\end{align}
\end{corollary}

By the Stieltjes--Perron inversion theorem 
\cite{Henrici:1974:v2,Schmudgen:2012:2012942602}, the matrix-valued spectral
measure $\bmu$ with components $\mu_{jk}$, $j,k=1,\ldots,d$, is given by the
weak limit $\d\bmu(\lambda)=-\lim_{\epsilon\downarrow0}\text{Im}(\mathbf{F}(\lambda+\imath\epsilon))(\d\lambda/\pi)$, i.e.
\begin{align}\label{eq:weak_limit_mu}
\int_0^1\xi(\lambda)\;\d\bmu(\lambda)
=-\frac{1}{\pi}\lim_{\epsilon\downarrow0}
\int_0^1\xi(\lambda)\;\text{Im}(\mathbf{F}(\lambda+\imath\epsilon))\, \d\lambda,
\end{align}
for all smooth scalar test functions $\xi(\lambda)$, where
$(\mathbf{F}(s))_{jk}=F_{jk}(s)$. From equation
\eqref{eq:weak_limit_mu} and the identities 
\begin{align}
m_{jk}(h)=h\,w_{jk}(z),
\quad
\tilde{m}_{jk}(h)=h\,\tilde{w}_{jk}(z),
\end{align}
which follow from equation~\eqref{eq:m_h}, it can be
shown~\cite{Murphy:JMP:063506} that the measures $\mu_{jk}$ and
$\alpha_{jk}$, and the measures $\eta_{jk}$ and $\kappa_{jk}$ are related by  
\begin{align}\label{eq:Measure_Relations}
&\lambda\alpha_{jk}(\lambda)=(1-\lambda)\mu_{jk}(1-\lambda) +
\lambda\,(\,m_{jk}(0)\delta_0(\d\lambda)+w_{jk}(0)(\lambda-1)\delta_1(\d\lambda)\,),
\\
&\lambda\kappa_{jk}(\lambda)=(1-\lambda)\eta_{jk}(1-\lambda) +
\lambda\,(\,\tilde{m}_{jk}(0)\delta_0(\d\lambda)+\tilde{w}_{jk}(0)(\lambda-1)\delta_1(\d\lambda)\,).
\notag     
\end{align}
Here, $m_{jk}(0)=m_{jk}(h)|_{h=0}$ and $w_{jk}(0)=w_{jk}(z)|_{z=0}$, for example, and
$\delta_a(\d\lambda)$ is the delta measure concentrated at $\lambda=a$. Equations 
\eqref{eq:Stieltjes_F} and~\eqref{eq:Stieltjes_G_X} demonstrate
the many symmetries between the functions $m_{jk}(h)$, $w_{jk}(z)$,
$\tilde{m}_{jk}(h)$, and $\tilde{w}_{jk}(z)$, and the respective 
measures $\mu_{jk}$, $\alpha_{jk}$, $\eta_{jk}$, and $\kappa_{jk}$.

\subsection{Sobolev space approach}
In this section we provide a rigorous mathematical foundation for an alternative approach 
\cite{Barabash:JPCM:10323,Milton:JAP-5294,Bergman:PRC-377:9,Bergman:1979:PRB:2359:19:4}, 
for providing Stieltjes integral representations for the bulk transport coefficients for composite 
materials. In this alternative approach, instead of utilizing the resolvent representations for the electric 
field $\vecE$ in \eqref{eq:Resolvent_representations_X}, a resolvent representation for the electric potential
$\phi$ is utilized, where $\vecE = \vecE_0 + \bnabla\phi$. This suggests use of the Hilbert space
\begin{align}
\Hs^{1,2}_0=\{f\in\Hs_0:\langle|\bnabla f|^2\rangle<\infty, \ \langle f\rangle=0\},
\end{align}
which is an example of a Sobolev space \cite{Folland:95:PDEs}. It is clear that $\Hs^{1,2}_0\subseteq\Hs^1_0$.
The following lemma 
\cite{Murphy:ADSTPF-2017} establishes that the Hilbert spaces $\Hs^{1,2}_0$ and $\Hs_\times$ 
are in one-to-one isometric correspondence.
\begin{lemma}\label{lem:isometry}
The Hilbert spaces $\Hs^{1,2}_0$ and $\Hs_\times$ are in one-to-one isometric correspondence. 
More specifically, temporarily denote the inner-product induced norm of the Hilbert space $\Hs^{1,2}_0$ by 
$\|f\|_{1,2} =\langle\bnabla f\bcdot\bnabla f\rangle^{1/2}$ and the inner-product induced norm of 
the Hilbert space $\Hs_\times$ by $\|\bpsi\|_\times =\langle\bpsi\bcdot\bpsi\rangle^{1/2}$.
Then, for every $f\in\Hs^{1,2}_0$ we have $\bnabla f\in\Hs_\times$ and $\|\bnabla f\|_\times=\|f\|_{1,2}$. 
Conversely, for every $\bpsi\in\Hs_\times$ there exists unique $f\in\Hs^{1,2}_0$ 
such that $\bpsi=\bnabla f$ and $\|f\|_{1,2}=\|\bpsi\|_\times$.
\end{lemma}
\begin{proof}
By Theorem \eqref{thm:H2_decomp_Hx_H.} and equation \eqref{eq:curl_grad_div_Ran_Ker_equality}, 
if $f \in \Hs^{1,2}_0$ then $\bnabla f \in\Rc({\bnabla})$ and $\Gamma\bnabla f = \bnabla f$. Moreover, 
$\|\bnabla f\|^2_\times=\langle\bnabla f\bcdot\bnabla f\rangle=\| f\|^2_{1,2} < \infty$, so that 
$\bnabla f\in\Hs_\times$. Consequently, for every 
$f \in \Hs^{1,2}_0$ we have $\bnabla f \in\Hs_\times$ and 
$\|\bnabla f\|^2_\times = \|f\|^2_{1,2}$. Conversely, $\bpsi \in\Hs_\times$ implies $\bpsi = \Gamma\bpsi = \bnabla f$, where we have defined the scalar-valued function $f=(\bnabla^*\bnabla)^{-1}\bnabla^*\bpsi$. Since
$\bpsi = \bnabla f$, the $\Hs^{1,2}_0$ norm of $f$ satisfies 
$\|f\|^2_{1,2}=\langle\bpsi\bcdot\bpsi\rangle=\|\bpsi\|^2_\times<\infty$, so that $f \in \Hs^{1,2}_0$. Moreover, $f$ is uniquely determined by $\bpsi$ (up to equivalence class \cite{Folland:99:RealAnalysis}), since if $f_1=(\bnabla^*\bnabla)^{-1}\bnabla^*\bpsi$ 
and $f_2=(\bnabla^*\bnabla)^{-1}\bnabla^*\bpsi$ then $\Gamma\bpsi=\bpsi$ implies that 
$\|f_1-f_2\|_{1,2} = \|\bpsi-\bpsi\|_\times=0$. Consequently, for every $\bpsi \in\Hs_\times$ there exists unique 
$f \in \Hs^{1,2}_0$ such that $\bpsi = \bnabla f$ and $\|f\|_{1,2} = \|\bpsi\|_\times$. In summary, the Hilbert spaces $\Hs^{1,2}_0$ and $\Hs_\times$ are in one-to-one isometric correspondence.
\end{proof}

We now establish that the potential $\phi$ has a resolvent representation involving a linear operator.
\begin{lemma}
\label{lem:resolvent_phi}
The potential $\phi$ has the following resolvent formula involving a linear operator $M$,
\begin{align}
    \label{eq:resolvent_phi}
    \phi=(s-M)^{-1}M\phi_0,
    \quad
    M=(\bnabla^*\bnabla)^{-1}[\bnabla^*X_1\bnabla]\,.
\end{align}
\end{lemma}
\begin{proof}
From Lemma \eqref{lem:isometry} we have 
$\vecE = \vecE_0 + \vecE_f$, where $\vecE_f=\Gamma\vecE_f=\bnabla\phi$ with  $\phi=(\bnabla^*\bnabla)^{-1}\bnabla^*\vecE_f$. 
We can write $\vecE_0=\bnabla\phi_0$ where
$\phi_0=\vecE_0\bcdot\vecx$. In the proof of Lemma \ref{lem:resolvent_X1E} which establishes a resolvent
representation for the field $X_1\vecE$, we utilized the formula 
$\Gamma\vecJ=0$, which can be obtained by applying the operator $\bnabla(\bnabla^*\bnabla)^{-1}$ to the 
formula $0=\bnabla\bcdot\vecJ=-\bnabla^*\vecJ$ in \eqref{eq:Maxwells_Equations_E0}. Instead, if 
we only apply the operator $(\bnabla^*\bnabla)^{-1}$ to the formula $\bnabla^*\vecJ=0$, and use equation \eqref{eq:sig_rho_s_t} to write $\vecJ=\bsigma\vecE=\sigma_2(1-X_1/s)(\vecE_0+\vecE_f)$, with $\vecE_0=\bnabla\phi_0$ and
$\vecE_f=\bnabla\phi$, using $\bnabla^*\vecE_0=0$ and rearrange, we obtain equation \eqref{eq:resolvent_phi}.
\end{proof}

Now, we use $\vecE = \bnabla\phi_0+\bnabla\phi$, equation \eqref{eq:sig_rho_s_t}, and $\langle\vecJ\bcdot\vecE\rangle=\langle\vecJ\bcdot\vecE_0\rangle$
to write the following functional formula for a diagonal component $\sigma^*_{jj}$, $j=1,\ldots,d$, of $\bsigma$ 
which is analogous to equation \eqref{eq:Energy_Reps} 
\begin{align}
\label{eq:functional_phi}
\sigma^*_{jj}E_0^2=
\sigma_2(E_0^2-\langle X_1\vecE\bcdot\vecE_0\rangle/s)=
\sigma_2(
E_0^2-
\langle X_1\vecE_0\bcdot\vecE_0\rangle/s-
\langle X_1\bnabla\phi\bcdot\bnabla\phi_0\rangle/s).
\end{align}
Equation \eqref{eq:functional_phi} indicates we should consider the inner-product 
$\langle f,g\rangle_{1,2}=\langle X_1\bnabla f\bcdot\bnabla g\rangle$, which is clearly a bi-sesquilinear 
form. However, it is not an inner-product unless \cite{Folland:99:RealAnalysis} we have 
$\langle f,f\rangle_{1,2}\in(0,\infty)$ for all non-zero $f$. This requires us to instead consider the Hilbert space 
\begin{align}
\Hs^{1,2}_{X_1}=\{f\in\Hs^{1,2}_0:X_1\bnabla f\neq0\}.
\end{align}

The properties of $M$ in \eqref{eq:resolvent_phi} with respect to the $\Hs^{1,2}_{X_1}$-inner-product follow 
from the properties of the operator $X_1\Gamma X_1$ with respect to the $\Hs^1$-inner-product, which we
establish in the following theorem.
\begin{theorem}
\label{thm:M_pos_bound_symm}
The operator $M$ with $\Dc(M)=\Hs^{1,2}_{X_1}$ is a positive, bounded, self-adjoint operator  
with spectrum $\Sigma\subseteq[0,1]$, satisfying
\begin{align}
    \label{eq:Mn_Gamma}
    \bnabla M^n=[\Gamma X_1]^n\bnabla\,.
\end{align}
\end{theorem}
\begin{proof}
From the definition of $\Gamma=\bnabla(\bnabla^*\bnabla)^{-1}\bnabla^*$ and
$M=(\bnabla^*\bnabla)^{-1}[\bnabla^*X_1\bnabla]$, it is clear that
\begin{align}\label{eq:M_Gamma}
    \bnabla M=\Gamma X_1\bnabla,
\end{align}
which implies that $\bnabla M^2=\Gamma X_1\bnabla M=[\Gamma X_1]^2\bnabla$. Consequently, iteratively applying the formula \eqref{eq:M_Gamma} establishes equation \eqref{eq:Mn_Gamma}.

Since the operator $X_1\Gamma X_1$ is positive on $\Hs^1$, by equation \eqref{eq:M_Gamma}
$M$ is also a positive operator on $\Hs^{1,2}_{X_1}$:
\begin{align}
    \langle M f,f\rangle_{1,2}=
    \langle X_1\Gamma X_1\bnabla f\bcdot\bnabla f\rangle
    \ge0.
\end{align}
Since the operator $X_1\Gamma X_1$ is bounded with $\|X_1\Gamma X_1\|\le1$ and $X_1^2=X_1$, 
by equation \eqref{eq:M_Gamma} the operator norm of $M$ is also 
bounded by 1:
\begin{align}
    \label{eq:M_Gamma_norm}
    \|M f\|_{1,2}=
    \|X_1\Gamma X_1\bnabla f \|\le
    \|X_1\Gamma X_1\|\|\bnabla f \|\le
    \|\bnabla f \|.
\end{align}
Since the operator $X_1\Gamma X_1$ is symmetric on $\Hs^1$ and $X_1^2=X_1$, 
by equation \eqref{eq:M_Gamma} $M$ is also symmetric with respect to the $\Hs^{1,2}_{X_1}$-inner-product:
\begin{align}
    \langle M f,f\rangle_{1,2}=
    \langle X_1\Gamma X_1\bnabla f\bcdot\bnabla f\rangle=
    \langle \bnabla f\bcdot X_1\Gamma X_1\bnabla f\rangle=
    \langle X_1\bnabla f\bcdot \Gamma X_1\bnabla f\rangle=
    \langle f,M f\rangle_{1,2}
\end{align}
The positive, bounded, symmetric operator $M$ has positive, bounded, self-adjoint extensions
\cite{Zoltzan:2023:03081087,Schmudgen:2012:2012942602} and can therefore be considered a self-adjoint
operator. Just as in Theorem \ref{thm:self-adjoint_X1GammaX1}, since $M$ is positive and bounded by one
in operator norm, the spectrum $\Sigma$ of $M$ satisfies $\Sigma\subseteq[0,1]$.
\end{proof}
\begin{theorem}
\label{thm:isometry_M_Mb}
Consider the operators $M=(\bnabla^*\bnabla)^{-1}[\bnabla^*X_1\bnabla]$ and 
$\bM=X_1\Gamma X_1$ with domains 
$\Dc(M)=\{\xi\in\Hs^{1,2}_0:X_1\bnabla \xi\neq0\}$ and $\Dc(\bM)=\{\bxi\in\Hs_\times:X_1\bxi\neq0 \}$. 
The domains $\Dc(M)$ and $\Dc(\bM)$ and ranges $\Rc(M)$ and $\Rc(\bM)$ are in one-to-one isometric 
correspondence. More specifically, temporarily denote the inner-product induced norm of the Hilbert space 
$\Hs^{1,2}_0$ by $\|f\|_{1,2} =\langle\bnabla f\bcdot\bnabla f\rangle^{1/2}$ and the inner-product induced 
norm of the Hilbert space $\Hs_\times$ by $\|\bpsi\|_\times =\langle\bpsi\bcdot\bpsi\rangle^{1/2}$.
Then, for every $f\in\Dc(M)$ we have $\bnabla f\in\Dc(\bM)$ with $\|\bnabla f\|_\times=\|f\|_{1,2}$ and $\|\bM\bnabla f\|_\times=\|M f\|_{1,2}$. 
Conversely, for every $\bpsi\in\Dc(\Mb)$ there exists unique $f\in\Hs^{1,2}_0$ 
such that $\bpsi=\bnabla f$, $f\in\Dc(M)$, and $\|Mf\|_{1,2}=\|\bM\bpsi\|_\times$.
\end{theorem}
\begin{proof}
Let $f\in\Dc(M)$ so that $f\in\Hs^{1,2}_0$ and $X_1\bnabla f\neq0$. By Lemma \ref{lem:isometry} we have
$\bnabla f\in\Hs_\times$ with $\|\bnabla f\|_\times=\|f\|_{1,2}$. Since we also have $X_1\bnabla f\neq0$ 
it follows that $\bnabla f\in\Dc(\bM)$. By equation \eqref{eq:M_Gamma_norm} we have 
$\|M f\|_{1,2}=\|\bM\bnabla f \|_\times$. Conversely, let $\bpsi\in\Dc(\bM)$ so that $\bpsi\in\Hs_\times$ and
$X_1\bpsi\neq0$. By Lemma \ref{lem:isometry}, there exists unique $f\in\Hs^{1,2}_0$ such that $\bpsi=\bnabla f$
and $\|\bpsi\|_\times=\|f\|_{1,2}$. Since $X_1\bnabla f\neq0$, we also have $f\in\Dc(M)$. 
By equation \eqref{eq:M_Gamma_norm} we have 
$\|\bM\bnabla f \|=\|M f\|_{1,2}$.

\end{proof}

We emphasize that adding the condition $X_1\bpsi\neq0$ to the definition of the domain of the 
operator $\bM=X_1\Gamma X_1$,
$\Dc(\bM)=\{\bxi\in\Hs_\times:X_1\bxi\neq0 \}$, is needed to remove the degenerate case where 
$\bxi\in\Hs_\times$ also satisfies $\bxi\in\Kc(\bM)$. We now look into details associated with this.
By Lemma \ref{lem:orthogonal} we have $\Dc(\bM)=\Rc(X_1)\oplus\Rc(X_1)^\perp$
and, since $X_1+X_2=I$, $\bpsi$ can be written uniquely as $\bpsi=X_1\bpsi+X_2\bpsi$, with 
$X_1\bpsi\in\Rc(X_1)$ and $X_2\bpsi\in\Rc(X_1)^\perp$. 
From Lemma \ref{eq:projection_decomp_sigma_poly} we have $X_i X_j=X_i\delta_{ij}$, $i,j=1,2$,
hence $\bM X_1\bpsi=\bM\bpsi$ and $\bM X_2\bpsi=0$. Consequently,  since $X_1\bpsi\neq0$, 
by equation \eqref{eq:M_Gamma_norm} we have $\|\bM\bnabla f \|=\|M f\|_{1,2}$.
\begin{align}
\|\bM(X_1\bpsi+X_2\bpsi) \|=
\|\bM X_1\bpsi \|=
\|\bM\bpsi \|=
\|\bM\bnabla f \|=
\|M f\|_{1,2}\,.
\end{align}
If $X_1\bpsi=0$ then we would instead have $0=\|\bM\bpsi \|=\|M f\|_{1,2}\,$.

\begin{theorem}
\label{thm:integral_rep_M}
Define the coordinate system so that $\bnabla\phi_0=\vecE_0=E_0\vece_j$, for some $j=1,\ldots,d$.
The diagonal component $\sigma^*_{jj}$ of the effective
conductivity matrix $\bsigma^*$ for a uniaxial polycrystalline medium has the Stieltjes integral representation
\begin{align}
    \label{eq:Bergman_integral_rearranged}
    \sigma^*_{jj}/\sigma_2=1-\int_0^1\frac{\d\nu_{jj}(\lambda)}{s-\lambda}\,,
\end{align}
where $\d\nu_{jj}(\lambda)=\d\langle \Qt(\lambda)\phi_0/E_0,\phi_0/E_0\rangle_{1,2}=\d\langle X_1\bnabla \Qt(\lambda)\phi_0/E_0\bcdot\bnabla\phi_0/E_0\rangle$, $\Qt(\lambda)$ is the resolution of the identity   
associated with the self-adjoint operator $M$ and $\nu$ is a spectral measure for $M$.
Moreover, let $\d\mu_{jj}(\lambda)=\d\langle Q(\lambda)X_1\vece_j,\vece_j\rangle$ be the spectral measure
in Theorem \ref{thm:X1E} associated with the self-adjoint operator $\bM=X_1\Gamma X_1$, then 
the measures $\nu_{jj}$ and $\mu_{jj}$ are identical,
\begin{align}
    \nu_{jj}\equiv\mu_{jj}\,.
\end{align}
Moreover the operators $Q(\lambda)$ and $\Qt(\lambda)$ are related by
\begin{align}
    \label{eq:Q_Qt}
    X_1(\bnabla \Qt(\lambda)-Q(\lambda)\bnabla)=0. 
\end{align}
\end{theorem}
\begin{proof}
Combining equations
\eqref{eq:resolvent_phi} and \eqref{eq:functional_phi} yields
\begin{align}
    \label{eq:functional_phi_X1}
    E_0^2\sigma^*_{jj}/\sigma_2=
    E_0^2-
    \langle X_1\vecE_0\bcdot\vecE_0\rangle/s-
    \langle (s-M)^{-1}M\phi_0,\phi_0\rangle_{1,2}/s\,.
\end{align}
By the spectral theorem formula in \eqref{eq:Spectral_Theorem} and equation \eqref{eq:functional_phi_X1} 
we have
\begin{align}
    \label{eq:Bergman_integral}
    \sigma^*_{jj}/\sigma_2=
    1-
    \frac{1}{sE_0^2}\langle X_1\vecE_0\bcdot\vecE_0\rangle-
    \frac{1}{s}\int_0^1\frac{\lambda\d\nu_{jj}(\lambda)}{s-\lambda}\,,
\end{align}
where $\d\nu_{jj}(\lambda)=\d\langle \Qt(\lambda)\phi_0/E_0,\phi_0/E_0\rangle_{1,2}=\d\langle X_1\bnabla \Qt(\lambda)\phi_0/E_0\bcdot\bnabla\phi_0/E_0\rangle$, $\Qt(\lambda)$ is the resolution of the identity operator in 
one-to-one correspondence with the self-adjoint operator $M$, and $\nu$ is a spectral measure for $M$
with measure mass given \cite{Stone:64} by 
\begin{align}
    \label{eq:measure_mass_M}
    \nu_{jj}^0=
    \langle\phi_0,\phi_0\rangle_{1,2}/E_0^2=
    \langle X_1\vecE_0\bcdot\vecE_0\rangle/E_0^2=
    \langle X_1\vece_j\bcdot\vece_j\rangle\,.
\end{align}
Consequently, since $\lambda/(s-\lambda)=-1+s/(s-\lambda)$, equation 
\eqref{eq:Bergman_integral} can be written as equation \eqref{eq:Bergman_integral_rearranged}.

We now show that the spectral measure $\nu_{jj}$ for the operator $M$ and the 
the spectral measure $\mu_{jj}$ for the operator $\bM=X_1\Gamma X_1$ 
are identical. From equation \eqref{eq:measure_mass_M} $\nu_{jj}^0=\mu_{jj}^0$. Applying equation
\eqref{eq:Mn_Gamma} and using 
$X_1^2=X_1$ and $\vecE_0=\bnabla\phi_0$ yields
\begin{align}
    \nu_{jj}^n=
    E_0^{-2}\langle M^n\phi_0,\phi_0\rangle_{1,2}=
    E_0^{-2}\langle X_1\bnabla M^n\phi_0\bcdot\bnabla\phi_0\rangle=
    E_0^{-2}\langle X_1[\Gamma X_1]^n\bnabla \phi_0\bcdot\bnabla\phi_0\rangle=
    \mu_{jj}^n\,.
\end{align}
Since $\nu_{jj}^n=\mu_{jj}^n$ for all $n=0,1,2,\ldots$, and the Hausdorff 
moment problem is \emph{determinate} \cite{Shohat:1963}, we have that $\nu_{jj}\equiv\mu_{jj}$, i.e., 
the measures $\nu_{jj}$ and $\mu_{jj}$ are identical.  Since $\vecE_0=\bnabla\phi_0$, it follows that 
\begin{align}
    0=
    \langle X_1\bnabla \Qt(\lambda)\phi_0\bcdot\bnabla\phi_0\rangle-
    \langle Q(\lambda)X_1\vecE_0\bcdot\vecE_0\rangle=
    \langle (X_1\bnabla \Qt(\lambda)-Q(\lambda)X_1\bnabla)\phi_0\bcdot\vecE_0\rangle.
\end{align}
Since the vector $\vecE_0$, hence $\phi_0$ is arbitrary, we have the following weak identity
\begin{align}
    X_1\bnabla \Qt(\lambda)=Q(\lambda)X_1\bnabla. 
\end{align}
From Theorem \ref{thm:X1E} the operators $Q(\lambda)$ and $X_1$ commute, which establishes equation
\eqref{eq:Q_Qt}.
\end{proof}

Clearly the results in this section also hold under the substitution $(X_1,s,\sigma_2)\mapsto(X_2,t,\sigma_1)$.
By Theorem \ref{thm:self-adjoint_two-component} we also have the following corollary.
\begin{corollary}
The results of Lemma \ref{lem:resolvent_phi}, Theorems \ref{thm:M_pos_bound_symm}, 
\ref{thm:isometry_M_Mb},  and \ref{thm:integral_rep_M} hold for the setting of a two-component composite material
under the substitution of $X_1$ with $\chi_1$. Consequently, the corresponding results developed in this section 
also hold under the substitution $(\chi_1,s,\sigma_2)\mapsto(\chi_2,t,\sigma_1)$.
\end{corollary}

\section{Bounding Procedure}\label{sec:Bounding_Procedure}
In this section, we review a mathematical framework which utilizes the Stieltjes integral representations 
for the diagonal components $\sigma^*_{kk}$ and $\rho^*_{kk}$, $k=1,\ldots,d$, of the effective conductivity 
$\bsigma^*$ and resistivity $\brho^*$ matrices in equation \eqref{eq:Stieltjes_G_X}, involving \emph{positive} 
measures $\mu_{kk}$ and $\eta_{kk}$, to provide rigorous bounds for $\sigma^*_{kk}$ and $\rho^*_{kk}$. The 
bounds incorporate known geometric information about the polycrystalline material, given in terms of the moments 
$\mu_{kk}^n$ and $\eta_{kk}^n$ of these measures in \eqref{eq:Moments_mu}. For example, 
$\mu_{kk}^n=\int_0^1\lambda^n \d\mu_{kk}(\lambda)=\langle[X_1\Gamma X_1]^n X_1\vece_k\bcdot\vece_k\rangle$ contains statistical information about the microgeometry of the polycrystalline material through the random projection matrix $X_1$, involving the distribution of the orientation angles of the composite crystals.

For example, in the case of  2D polycrystalline media, $d=2$, equation
\eqref{eq:polycrystal_parameters_2D} implies that 
\begin{align}\label{eq:measure_mass_theta_2D}
\mu_{11}^0=\langle\cos^2\theta\rangle, 
\quad
\mu_{22}^0=\langle\sin^2\theta\rangle. 
\end{align}
If the material is isotropic and $\theta(0 ,\omega)=\theta^{\,\prime}(\omega)$ is uniformly distributed on $[0,2\pi]$ then 
the averages in equation \eqref{eq:measure_mass_theta_2D} yield 
\begin{align}\label{eq:measure_mass_theta_2D_isotropic}
\mu_{11}^0=\mu_{22}^0=1/2.
\end{align}
In the case of  3D polycrystalline media, defining $R=R_1 R_2 R_3$, where $R_1$, $R_2$, and $R_3$ are the 3D rotation matrices in equation \eqref{eq:polycrystal_parameters_3D}, yields 
\begin{align}\label{eq:measure_mass_theta_3D}
\mu_{11}^0=\langle\cos^2{\theta_2}\cos^2{\theta_3}\rangle, 
\quad
\mu_{22}^0=\langle\cos^2{\theta_2}\sin^2{\theta_3}\rangle, 
\quad
\mu_{33}^0=\langle\sin^2{\theta_2}\rangle. 	
\end{align}
If the material is isotropic and the $\theta_i(0 ,\omega)=\theta_i^{\,\prime}(\omega)$ are uniformly distributed on $[0,2\pi]$ then 
the averages in equation \eqref{eq:measure_mass_theta_3D} yield 
\begin{align}\label{eq:measure_mass_theta_3D_isotropic}
\mu_{11}^0=\mu_{22}^0=1/4,
\quad 
\mu_{33}^0=1/2.
\end{align}

Transforming the integrals in \eqref{eq:Stieltjes_G_X} in terms of Stieltjes functions of the parameters $h$ and $z=1/h$ \cite{Murphy:JMP:063506} enables the theory of Pad\'{e} approximants to be applied, which provides bounds for $\sigma^*_{kk}$ and $\rho^*_{kk}$ in terms of the 
approximants \cite{Murphy:2024:arXiv:Spectral_Measures_Iterative_Bounds}. An alternate approach fixes the
contrast parameter $s$ and varies over admissible sets of measures $\mu_{kk}$ and 
$\eta_{kk}$~\cite{Golden:CMP-473,Golden:JMPS-333:4}. 
When $h$ is real-valued, $\sigma^*_{kk}$ is bounded by a real interval, and when $h$ is 
complex valued  $\sigma^*_{kk}$ is bounded by arcs of circles, and similarly for $\rho^*_{kk}$   \cite{Baker:1996:Book:Pade,Golden:CMP-473,Golden:JMPS-333:4}. 
The bounding region becomes progressively smaller as more moments are
known~\cite{Milton:JAP-5294,Golden:JMPS-333:4,Baker:1996:Book:Pade}. Since the bounding
procedure associated with the functions $G_{kk}(t)$ and $H_{kk}(t)$
in~\eqref{eq:Stieltjes_F} is analogous, we will focus on that
involving $F_{kk}(s)$ and $E_{kk}(s)$.

The bounds for $\sigma_{kk}^*$ and $\rho^*_{kk}$ follow from three important
properties of the functions $F_{kk}(s)$ and $E_{kk}(s)$. First, their
integral representations displayed in equation~\eqref{eq:Stieltjes_G_X}
\emph{separate} parameter information in $s$ and $E_0$ from the 
geometry of the composite, which is encoded in the underlying spectral
measures $\mu_{kk}$ and $\eta_{kk}$ via their moments $\mu_{kk}^n$ and
$\eta_{kk}^n$, $n=0,1,2,\ldots$~\cite{Bruno:JSP-365,Golden:CMP-473}. Second,
these integral representations are \emph{linear} functionals of the
spectral measures. Finally, $\mu_{kk}$ and $\eta_{kk}$ are \emph{positive}
measures, in contrast to $\mu_{jk}$ and $\eta_{jk}$ for $j\neq k$ which are 
\emph{signed} measures. In this
section, we review how these three properties yield rigorous bounds
for the diagonal components of the effective parameters $\sigma^*_{kk}$ and
$\rho^*_{kk}$~\cite{Golden:CMP-473,Golden:JMPS-333:4}.

For simplicity, we will focus on one diagonal component $\sigma^*_{kk}$ and
$\rho^*_{kk}$ of the effective conductivity and resistivity tensors
$\bsig^*$ and $\brho^*$, for some $k=1,\ldots,d$, and set $\sigma^*=\sigma_{kk}^*$,
$F(s)=F_{kk}(s)$, $m(h)=m_{kk}(h)$, $\mu=\mu_{kk}$, $E(s)=E_{kk}(s)$,
$\tilde{m}(h)=\tilde{m}_{kk}(h)$, and $\eta=\eta_{kk}$. 
Here,
$F(s)=1-m(h)$ and $E(s)=1-\tilde{m}(h)$.
We will also exploit the symmetries between
$F(s)$ and $E(s)$ in equation~\eqref{eq:Stieltjes_G_X} and initially 
focus on the function $F(s)$ and the measure $\mu$, referring to the
function $E(s)$ and the measure $\eta$ where appropriate.

Bounds for $\sigma^*$ are obtained as follows, while those for $\rho^*$ are
obtained analogously. The support of the measure 
$\mu$ is contained in the interval $[0,1]$ and its mass $\mu^0$ is given in equations
\eqref{eq:measure_mass_theta_2D}--\eqref{eq:measure_mass_theta_3D_isotropic}, where $0\leq \mu^0\leq1$. Consider the set $\mathscr{M}$ of
positive Borel measures on $[0,1]$ with mass less than or equal to 1. By
equation~\eqref{eq:Stieltjes_F}, for fixed $s\in\mathbb{C}\backslash[0,1]$,
$F(s)$ is a linear functional of the measure $\mu$,
$F:\mathscr{M}\mapsto\mathbb{C}$, and we write $F(s)=F(s,\mu)$ and
$m(h)=m(h,\mu)$. Suppose that we know the moments $\mu^n$ of the measure
$\mu$ for $n=0,\ldots,J$. Define the set $\mathscr{M}_J^\mu\subset\mathscr{M}$ of
measures by 
\begin{align}\label{eq:Measure_Set}
\mathscr{M}_J^\mu
=\left\{\nu\in\mathscr{M} \ \Big| \   \int_0^1\lambda^n\d\nu(\lambda)=\mu^n, \  n=0,\ldots,J\right\}. 
\end{align}
The set $A_J^\mu\subset\mathbb{C}$ that represents the possible
values of $m(h,\mu)=1-F(s,\mu)$ which is compatible with the
known information about the random medium is given by
\begin{align}\label{eq:Bounding_Set}
A_J^\mu
=\left\{\ m(h,\mu)\in\mathbb{C} \ | \
\ h\not\in(-\infty,0], \ \mu\in \mathscr{M}_J^\mu\right\}. 
\end{align}

The set of measures $\mathscr{M}_J^\mu$ is a compact, convex
subset of $\mathscr{M}$ with the topology of weak
convergence~\cite{Golden:CMP-473}. Since the mapping $F(s,\mu)$
in~\eqref{eq:Stieltjes_G_X} is linear in $\mu$, it follows that 
$A_J^\mu$ is a compact convex subset of the complex plane
$\mathbb{C}$. The extreme points of $\mathscr{M}_0^\mu$ are the one 
point measures $a\delta_b$, $0\leq a,b\leq1$~\cite{Dunford_Schwartz:LinOp_PtI},
while the extreme points of $\mathscr{M}_J^\mu$ for $J>0$ are weak limits
of convex combinations of measures of the
form~\cite{Karlin_Studden:Book:1966,Golden:CMP-473}   
\begin{align}\label{eq:Discrete_Measure}
\mu_J(\d\lambda)=\sum_{i=1}^{J+1}a_i\delta_{b_i}(\d\lambda), \quad
a_i\geq0, \quad 0\leq b_1<\cdots<b_{J+1}<1, \quad
\sum_{i=1}^{J+1}a_ib_i^n=\mu^n,
\end{align}
for $n=0,1,\ldots,J$. 
In general~\cite{Golden:CMP-473},  not every measure
$\mu\in\mathscr{M}_J^\mu$ gives rise to such a function $m(h,\mu)$. Therefore,
the set $A_J^\mu$ will \emph{contain} the exact range of values of the
effective conductivity~\cite{Golden:CMP-473}. This is sufficient for
the bounding procedure discussed in this section.

\begin{figure}                \centering
\includegraphics[width=\textwidth]{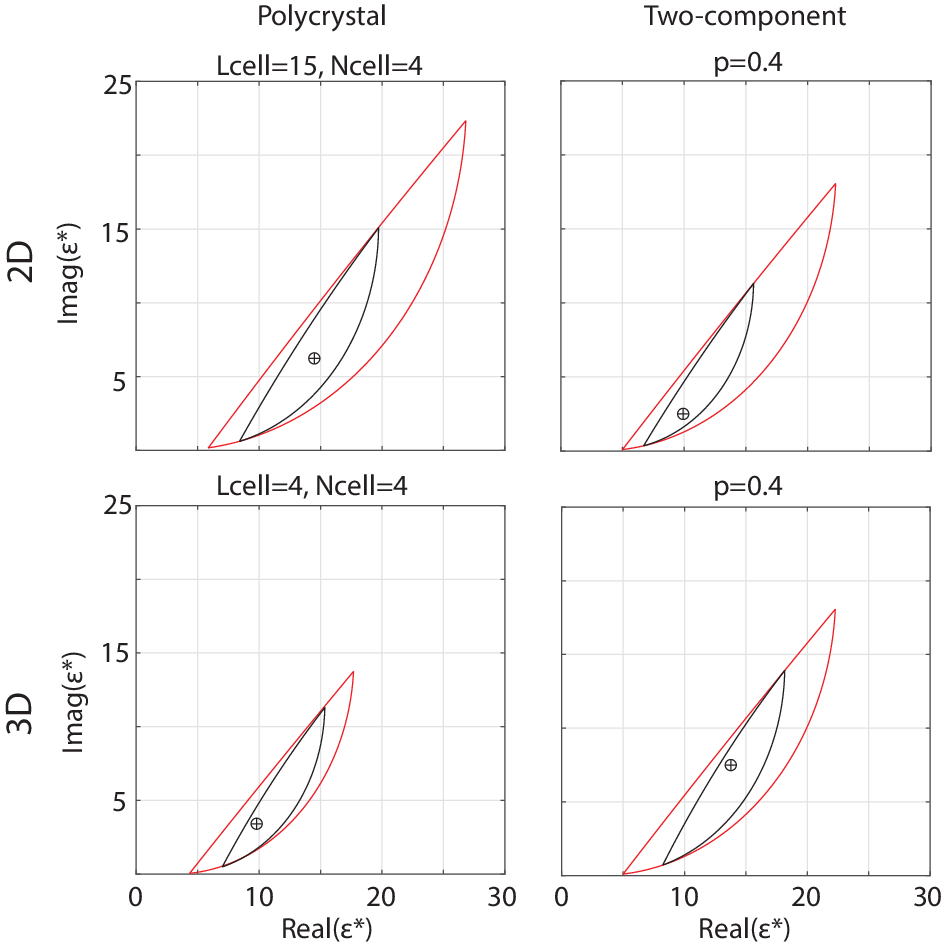}
\caption{\textbf{Bounds on the effective conductivity}. First order bounds (red curves) and second order, isotropic bounds (black curves) of the effective complex conductivity $\bsig^*_{kk}$ for discretized two-dimensional (top row) and three-dimensional (bottom row) polycrystalline and two-component materials. The first order bounds, which utilize the zeroth moment $\mu^0$, are given in \eqref{eq:0th_order_Bounds}. By utilizing the additonal geometric information of the first moment $\mu_1$, the second order bounds are significantly tighter than the first order bounds. Alternatively, the same bounds can be derived without $\mu^1$ by assuming isotropic material geometry \cite{Gully:PRSA:471:2174}. The collocated circles `o' and crosses `+' denote the values $\bsig^*_{11}$ computed in two equivalent ways given in Corollary \eqref{eq:Stieltjes_parameters} using the spectral measures $\text{d}\mu_{11}$ (circles) and $\text{d}\alpha_{11}$ (crosses)  corresponding to the random operators $X_1 \Gamma X_1$ and $X_2 \Gamma X_2$, respectively. Both the polycrystalline and two-component material geometry are square regions (cubic for 3D) and isotropic, with the discretized spatial side lengths set to $L=60$ for 2D and $L=16$ for 3D. For the 2D polycrystal, the spatial domain is divided into a $4 \times 4$ grid of square crystallites and the 3D polycrystal consists of a $4 \times 4 \times 4$ grid of cubic crystallites. Both two-component materials have a phase 1 volume fraction $p=0.4$. The conductivity $\sigma_1=51.074+i45.160$ is assigned to the $x$-direction for the individual crystallites as well as the first phase of the two-component geometry, and the conductivity $\sigma_2 = 3.070 + i0.0019$ is assigned to the $y$ and $z$ directions of the individual crystallites as well as the second phase of the two-component geometry. Details on the direct computation of spectral measures and discretization of the operators $X_1 \Gamma X_1$ and $X_2 \Gamma X_2$ will be the subject of a future publication. }


\label{polycrystal_bounds}

\end{figure}

By the symmetries between the formulas in
equation~\eqref{eq:Stieltjes_G_X}, the support of the measure $\eta$ is
contained in 
the interval $[0,1]$ and its mass is given by $\eta^0=1-\mu^0$, where
$0\leq \eta^0\leq1$. We can therefore define compact, convex sets
$\mathscr{M}_J^\eta\subset\mathscr{M}$ and $A_J^\eta\subset\mathbb{C}$ which are
analogous to those defined in equations~\eqref{eq:Measure_Set}
and~\eqref{eq:Bounding_Set}, respectively, involving the function 
$\tilde{m}(h,\eta)=1-E(s,\eta)$. Moreover, the extreme points of
$\mathscr{M}_0^\eta$ are the one point measures $c\delta_d$, $0\leq c,d\leq1$, 
while the extreme points of $\mathscr{M}_J^\eta$ are weak limits
of convex combinations of measures of the form given in
equation~\eqref{eq:Discrete_Measure}.

Consequently, in order to determine the extreme
points of the sets $A_J^\mu$ and $A_J^\eta$, it suffices to determine the
range of values in $\mathbb{C}$ of the functions $m(h,\mu_J)=1-F(s,\mu_J)$
and $\tilde{m}(h,\eta_J)=1-E(s,\eta_J)$, respectively, where  
\begin{align}\label{eq:Discrete_mh}
	F(s,\mu_J)=\sum_{i=1}^{J+1}\frac{a_i}{s-b_i}\,, \qquad
	E(s,\eta_J)=\sum_{i=1}^{J+1}\frac{c_i}{s-d_i}\,,
\end{align}
as the $a_i$, $b_i$, $c_i$, and $d_i$ vary under the
constraints given in equation ~\eqref{eq:Discrete_Measure}. While
$F(s,\mu_J)$ and $E(s,\eta_J)$ in~\eqref{eq:Discrete_mh} may not run over
all points in $A_J^\mu$ and 
$A_J^\eta$ as these parameters vary, they run over the
extreme points of these sets, which is sufficient due to their
convexity. It is important to note that, as the effective complex
conductivity $\sigma^*$ is given by $\sigma^*=\sigma_2m(h,\mu)=\sigma_1/\tilde{m}(h,\eta)$, the
regions $A_J^\mu$ and $A_J^\eta$ have to be mapped to the common
$\sigma^*$-plane to provide bounds for $\sigma^*$.

Consider the case where $J=0$  in~\eqref{eq:Discrete_mh} and the
distribution of the crystal orientation angles is fixed, so that $F(s,\mu_J)=\mu^0/(s-\lambda)$ and
$E(s,\eta_J)=\eta^0/(s-\tilde{\lambda})$. By the above discussion, the values of 
$F(s,\mu$) and $E(s,\eta)$ lie inside the circles $C_0(\lambda)$ and
$\tilde{C}_0(\tilde{\lambda})$, respectively, given by  
\begin{align}\label{eq:0th_order_Bounds}
	C_0(\lambda)=\frac{\mu^0}{s-\lambda}\,, \quad -\infty\leq\lambda\leq \infty, \qquad
	\tilde{C}_0(\tilde{\lambda})=\frac{\eta^0}{s-\tilde{\lambda}}\,, \quad
	-\infty\leq\tilde{\lambda}\leq \infty. 
\end{align}
In the $\sigma^*$-plane, the intersection of these two regions is bounded by
two circular arcs corresponding to $0\leq\lambda\leq \eta^0$ and $0\leq\tilde{\lambda}\leq \mu^0$
in~\eqref{eq:0th_order_Bounds}, and the values of $\sigma^*$ lie inside
this region~\cite{Golden:JMPS-333:4}. 
The arcs are traced out
as the aspect ratio varies. When the value of the component
conductivities $\sigma_1$ and $\sigma_2$ are real and positive, the bounding
region collapses to the interval
$1/(\mu^0/\sigma_1+\eta^0/\sigma_2)\leq\sigma^*\leq \mu^0\sigma_1+\eta^0\sigma_2$, which are the Wiener
bounds. 

Now consider the case where $J=1$ in~\eqref{eq:Discrete_mh}. 
If the random medium is also known to be statistically
isotropic, so that $\bsig^*$ is diagonal~\cite{Milton:2002:TC}, the first moments $\mu^1$ 
is known to be given by~\cite{Gully:PRSA:471:2174}   
\begin{align}\label{eq:First_Moments}
	\mu^1=\frac{d-1}{d^3}\,, \qquad
\end{align}
This procedure can then be repeated using methods of complex analysis to compute
bounds associated with isotropic materials shown in Figure \ref{polycrystal_bounds}
\cite{Golden:JMPS-333:4,Gully:PRSA:471:2174}.

\section{Conclusions}
We formulated the first rigorous mathematical framework that provides Stieltjes integral representations
for the bulk transport coefficients for uniaxial polycrystalline materials, involving spectral measures of self-adjoint
random operators of the form $X_1\Gamma X1$ and $X_1\Upsilon X_1$, where $\Gamma$ and $\Upsilon$
are non-random projections onto the range of generalized gradient and curl operators, and $X_1$ is a random
projection operator. We accomplished this by providing a detailed analysis, that establishes these operators 
are well defined and indeed self-adjoint on an appropriate Hilbert space. We show that the mathematical
framework describing effective transport for uniaxial polycrystalline materials is a direct analogue to the mathematical framework describing effective transport for locally isotropic two-component materials,
where one mathematical framework maps to the other by simply swapping the random projection operator
$X_1$ with the characteristic function $\chi_1$ associated with two-component media, which is also a random
projection operator. Consequently, we also simultaneously formulated the first rigorous mathematical framework that provides Stieltjes integral representations for the bulk transport coefficients for two-component random media, involving spectral measures of self-adjoint random operators. We provided an abstract formulation of the
Helmholtz theorem which led to the resolution of the identity $\Gamma+\Upsilon=I$, which enabled these results to be extended to inverses of effective parameters, e.g., effective conductivity and resistivity. We briefly reviewed rigorous bounds that follow from such Stieltjes integrals and partial knowledge about the material geometry validated the bounds by numerical calculations of the effective parameters for polycrystalline media.

\medskip

{\bf Acknowledgements.}
We gratefully acknowledge support from the 
Division of Mathematical Sciences at the US National Science 
Foundation (NSF) through Grants DMS-0940249, DMS-1413454,
DMS-1715680, DMS-2136198, and DMS-2206171.
We are also grateful for support from the Applied and 
Computational Analysis Program 
and the Arctic and Global Prediction Program 
at the US Office of Naval 
Research through grants 
N00014-13-1-0291,
N00014-18-1-2552,
N00014-18-1-2041
and
N00014-21-1-2909.

\medskip

\bibliographystyle{plain}
\bibliography{murphy}
\end{document}